\documentclass{article} 
\usepackage[utf8]{inputenc}
\usepackage[T1]{fontenc}
\usepackage{amsmath}
\usepackage{amsfonts}
\usepackage[margin=1in]{geometry}
\usepackage{mathrsfs}
\usepackage{float}
\usepackage{graphicx}
\usepackage{makecell}
\usepackage{thm-restate}
\usepackage{csquotes}
\usepackage{quantikz}

\input{preamble}
\usepackage{graphicx}
\usepackage{hyperref}
\usepackage{cleveref}
\usepackage{authblk} 
\newcommand{\coloneqq}{:=}

\theoremstyle{plain}

\theoremstyle{plain}
\newtheorem{innercustomprop}{Proposition}
\newenvironment{customprop}[1]
  {\renewcommand\theinnercustomprop{#1}\innercustomprop}
  {\endinnercustomprop}

\newtheorem{innercustomthm}{Theorem}
\newenvironment{customthm}[1]
  {\renewcommand\theinnercustomthm{#1}\innercustomthm}
  {\endinnercustomthm}
 
\title{Good Stabilizer Codes from Shallow Clifford Circuits \\ with Random Matchings}
\date{}
 
\author[1]{Emile Anand}
\author[2]{Elia Gorokhovsky}
\author[3]{Jennifer Hritz}
\author[4]{Jingtong Sun}

\affil[1]{Georgia Institute of Technology, School of Computer Science}
\affil[2]{Harvard University, Department of Mathematics}
\affil[3]{University of New Mexico, Department of Physics and Astronomy}
\affil[4]{California Institute of Technology}
\newcommand{\nocontentsline}[3]{}
\let\origtoc=\addcontentsline
\let\addcontentsline=\nocontentsline

\begin{document}
\maketitle

\begin{abstract}
Encoding quantum information with low circuit overhead is a fundamental challenge in fault-tolerant quantum computation. Random circuits provide a natural mechanism for rapidly spreading logical information through simple gates applied in parallel. Brown and Fawzi \cite{Brown_2013} showed that random Clifford circuits on two-qubit Clifford gates provide such encoders that achieve the quantum Gilbert-Varshamov rate-distance tradeoff with depth $O(\log^3 n)$. We show that the same asymptotic tradeoff is attained in optimal $O(\log n)$ depth under a gate distribution with a more restricted support. Specifically, for every fixed $\delta>0$ and sufficiently large $n$, if $\frac kn < 1 - H(\frac{d}{n}) - \frac{d}{n}\log_2 3 - \delta$, we can construct random circuits of depth $O(\log n)$ which define, with high probability, an $[n,k]$ stabilizer code of distance at least $d+1$, which matches the $\Omega(\log n)$ light-cone lower bound for linear distance encoders.
    
Our ensemble employs a \emph{random matching circuit} architecture consisting of $T$ independent permutation-invariant layers. In each layer, the qubits are paired up by a uniformly random perfect matching, and a random independent two-qubit Clifford gate is applied to each pair. The gate distribution need not be uniform over, or even have full support on, the two-qubit Clifford group; instead, we allow for very general distributions on Clifford gates satisfying three regularity conditions. In particular, the construction can be implemented using $n/2$ CNOT gates on randomly matched pairs in each layer, together with parallel one-qubit Clifford twirls. The regularity conditions we assume allow us to reduce the second-moment dynamics of our random circuits to a reversible Markov chain on binary support strings. We establish logarithmic hitting-time bounds for this Markov chain and comparisons of its stationary distribution to prove the coding properties of the circuits.
\end{abstract}

\newpage

\tableofcontents

\let\addcontentsline=\origtoc

\newpage

\section{Introduction}
Quantum error-correcting codes protect logical information from physical noise and are a fundamental ingredient of large-scale fault-tolerant quantum computation \cite{aharonov1997fault,gottesman1997stabilizer,gottesman2010introduction, egan2021fault, gong2022experimental, postler2022demonstration, zhao2022realization, gottesman2013fault}. The existence of asymptotically good quantum codes was established through probabilistic and algebraic methods, beginning with the CSS constructions \cite{Calderbank_1996} and the stabilizer-code framework based on additive codes over $\mathbb F_4$ \cite{calderbank1997quantumerrorcorrectioncodes}; but these probabilistic existence results do not by themselves provide a shallow or gate-constrained encoder. Their usefulness, however, also depends on the complexity of the encoding circuit, as a code with strong parameters may still require a circuit with large depth or complicated gate set. 

There has been a recent resurgence of interest in random (stabilizer) quantum error correction codes. Brown and Fawzi \cite{Brown_2013} showed that shallow random Clifford circuits of depth $O(\log^3 n)$ using $O(n\log^2 n)$ two-qubit gates in an all-to-all connected architecture form good quantum error correction codes. Their construction samples each two-qubit gate uniformly from the full two-qubit Clifford group $\mathcal C_2$. Specifically, they showed that random Clifford circuits of size $O(n \log^2 n)$ can encode $k$ logical qubits into $n$ physical qubits with distance $d$, achieving the quantum Gilbert-Varshamov rate-distance tradeoff. Although every element of $\mathcal{C}_2$ has a constant-size decomposition into elementary Clifford gates, sampling from the full two-qubit Clifford group requires a richer two-qubit gate ensemble and may require multiple native entangling operations per sampled gate. In fact, an open question, posed explicitly by Brown and Fawzi, is whether a more restricted and hardware-friendly gate set, particularly, the elementary Clifford generators (CNOT, Hadamard, and phase gates), would be sufficient.  More recently, \cite{wills2026lineartimeencodabledecodablequantum} constructs randomized and explicit families of asymptotically good CSS codes at any prescribed constant rate. Their encoders and inverse encoders use only CNOT gates, have $O(\log n)$ depth, and contain $O(n)$ gates to achieve linear distance. However, their exact rate-distance tradeoff remains undetermined, and they explicitly ask for a characterization of this tradeoff.

A parallel line of research studies random-circuit codes under geometric locality constraints, particularly in one-dimensional brickwork architectures \cite{PhysRevX.11.031066,Darmawan_2024,Liu_2026,Kroll:2026max}. These works establish different guarantees, such as \emph{approximate} quantum error correction or performance against specified noise channels, with logarithmic depth. In contrast, we study the exact minimum distance in an all-to-all architecture. These two settings are complementary to each other. Geometric locality imposes rate-distance tradeoffs; for example, two-dimensional local commuting-projector codes satisfy $kd^2=O(n)$ \cite{PhysRevLett.104.050503}. In contrast, the all-to-all architecture is not subject to this particular locality tradeoff, motivating the question below:
\begin{center}
    \emph{Can the quantum-GV tradeoff be attained by an encoder of optimal $O(\log n)$ depth, without additional work ancillas, and implemented using only a small elementary gate set?}
\end{center}

In this work, we answer the question affirmatively by analyzing the encoding power of random circuits with restricted two-qubit gate sets in an all-to-all architecture. For every fixed $\delta,m>0$ and all sufficiently large even $n$, we construct an ensemble of random matching circuits such that, whenever $\frac kn < 1 - H(\frac{d}{n}) - \frac{d}{n}\log_2 3 - \delta$, a circuit sampled from the ensemble has depth $O(\log n)$ and defines an $[n,k]$ stabilizer code of distance at least $d+1$ with probability at least $1-n^{-m}-2^{-\Omega_\delta(n)}$, where $H$ is the binary entropy function. In our principal construction, each matching round samples a uniformly random perfect matching, applies $n/2$ CNOT gates to the matched pairs in parallel, and applies independent one-qubit Clifford twirls before and after the CNOT layer with independent one-qubit Clifford twirls which may be supported on the three-element set $\{\mathsf I, \mathsf H \mathsf S,(\mathsf H \mathsf S)^2\}$. The resulting circuits use $O(n\log n)$ CNOT gates and $O(n\log n)$ one-qubit Clifford gates with optimal depth, thereby approaching the quantum Gilbert-Varshamov tradeoff.

More generally, the conclusion holds for any fixed distribution $p_2$ on two-qubit Clifford gates satisfying three regularity conditions: first, $p_2$ should be invariant under left and right multiplication by a subgroup of local Clifford gates that acts transitively on the non-identity single-qubit Paulis. Second, $p_2$ must assign positive probability to an entangling gate. Third, $p_2$ should be invariant under taking inverses. Under these regularity conditions, the resulting random matching circuits satisfy the same depth, distance, and success-probability guarantees. In particular, the locally twirled CNOT distribution satisfies these conditions. More generally, so does the locally twirled distribution obtained by choosing equiprobably between any fixed entangling two-qubit Clifford gate and its inverse. For a Clifford encoder, our regularity conditions reduce the coding problem to bounding, for each relevant nontrivial input Pauli, the probability that it is mapped to a Pauli of weight at most $d$. The local bi-invariance of $p_2$ makes the second-moment dynamics depend only on the binary support of the Pauli operator, yielding a reversible Markov chain on nonzero support strings. The induced chain on the Hamming weight $w$ has stationary distribution $\pi(w)={3^w\binom{n}{w}}/{(4^n-1)}$. We prove that low-weight states reach linear weight within $O(\log n)$ steps and that, at linear weight, the chain rapidly hits a neighborhood of the stationary weight $3n/4$. We then use reversibility to compare the probability of returning to low weight with the corresponding stationary tail, which yields the uniform low-weight estimates needed for the distance union bound.

There is a line of works on unitary $t$-designs which are ensembles whose $t$-fold moment operator matches that of the Haar measure exactly or up to some approximation error \cite{Dankert_2009}. \cite{Harrow_2009} show that polynomial-length random circuits give approximate $1$ and $2$-designs, and \cite{Brand_o_2016} proved that one-dimensional nearest-neighbor random circuits form approximate unitary $t$-designs for general $t$. Later work improves design depths, including \cite{haferkamp2022random}'s $O(nt^{5+o(1)})$-depth, and \cite{schuster2025random}'s log-depth design constructions. While exact unitary design and its sufficiently accurate approximations are a natural sufficient condition for Haar-like second moment behavior (indeed, random circuits are often analyzed through the lens of unitary designs), they are not necessary for obtaining GV-distance stabilizer codes. Our coding argument does not require convergence of the full second-moment channel, as it is enough that every nontrivial input Pauli have a very small probability of landing in the set of low-weight Paulis to control the code union bound, thereby allowing us to recover the quantum GV rate-distance tradeoff from a significantly more restricted random circuit ensemble with shallow depth.

Beyond quantum coding, random circuits play important roles in randomized benchmarking \cite{heinrich2022randomized,helsen2022general, knill2008randomized}, unitary-design constructions \cite{Brand_o_2016,haferkamp2022random,schuster2025random}, randomized measurements and classical-shadow estimation \cite{bertoni2024shallow}, and models of many-body dynamics \cite{nahum2017quantum}. These connections provide broader motivation for understanding Pauli spreading in restricted random-circuit ensembles.

\subsection{Preliminaries and Notation}

\paragraph{Notation.} For $n\in\mathbb N$, we write $[n]\coloneqq \{1,2,\dots,n\}$, and let $0^n$ denote the all-zero string of length $n$. All logarithms are in base $2$ unless explicitly stated otherwise. For $0\leq p \leq 1$, we let $H(p) = -p\log p - (1-p)\log (1-p)$ denote the binary entropy function, with the convention $0 \log 0 = 0$. For $x\in \{0,1\}^n$, we write the support of $x$ as $\mathrm{supp}(x) \coloneqq \{i\in[n] \mid x_i=1\}$, with $w(x) = |\mathrm{supp}(x)|$ for its Hamming weight. Finally, we use $\mathbbm{1}\{\cdot\}$ as the indicator function.

\paragraph{Pauli operators.} We work on the $n$-qubit Hilbert space $(\mathbb C^2)^{\otimes n}$ and use the Pauli basis to decompose operators. The single-qubit Pauli operators are
\begin{align*}
    \sigma_0 = \mathbf I = \begin{pmatrix}
        1 & 0 \\ 
        0 & 1
    \end{pmatrix}, \qquad \sigma_1 = X = \begin{pmatrix}
        0 & 1 \\ 
        1 & 0 
    \end{pmatrix}, \qquad
    \sigma_2 = Y = \begin{pmatrix}
        0 & -i \\ 
        i & 0
    \end{pmatrix},\qquad \sigma_3 =  Z = \begin{pmatrix}
        1 & 0 \\ 
        0 & -1 
    \end{pmatrix}
\end{align*}
For a Pauli index string $\nu \in \{0,1,2,3\}^n$, we define $\sigma_\nu = \sigma_{\nu_1} \otimes \cdots \otimes \sigma_{\nu_n}$. The support $\mathrm{supp}(\nu)$ of $\nu$ is then the subset $\{i\in[n] \mid \nu_i \neq 0\}$ and the weight is $w(\nu) = |\mathrm{supp}(\nu)|$. We define a string $b(\nu) \in \{0, 1\}^n$ as the unique binary string with $\mathrm{supp}(b(\nu)) = \mathrm{supp}(\nu)$. Since the operators $\{\sigma_\nu \mid \nu\in\{0,1,2,3\}^n\}$ form an orthogonal basis for the space of linear operators on $(\mathbb C^2)^{\otimes n}$ with respect to the Hilbert-Schmidt inner product ($\tr(\sigma_\mu \sigma_\nu) = 2^n \mathbbm{1} \{\mu=\nu\}$), we can decompose any operator $A$ acting on $(\mathbb{C}^2)^{\otimes n}$ as
$A = 2^{-n} \sum_{\nu \in \{0,1,2,3\}^n}\tr[\sigma_\nu A]\sigma_\nu$.

\paragraph{CNOT and SWAP gates.} The controlled NOT (CNOT) gate is a quantum logic gate which acts on two qubits: $\operatorname{CNOT}_{1\to2}|a,b\rangle=|a,a\oplus b\rangle$. Similarly, a SWAP gate is a quantum logic gate that exchanges the states of two qubits: $\operatorname{SWAP}|a,b\rangle=|b,a\rangle$.
  
\paragraph{Pauli and Clifford groups.} Let $\mathcal P_n$ denote the $n$-qubit Pauli group including phases, $\mathcal P_n \coloneqq \{\pm 1, \pm i\}\cdot \{\mI, X, Y, Z\}^{\otimes n}$. The $n$-qubit \textit{projective }Clifford group is the normalizer of $\mathcal P_n$ in the unitary group \textit{up to global phase}: $\mathcal C_n \coloneqq \{U \mid U\mathcal P_n U^\dagger = \mathcal P_n\}/U(1)$. Thus, for every $U\in \mathcal C_n$ and every Pauli string $\sigma_\nu$, there exists a phase $\omega\in \{\pm 1\}$ and a Pauli index $\mu\in \{0,1,2,3\}^n$ such that $U\sigma_\nu U^\dagger = \omega \sigma_\mu$. The single-qubit Clifford group is denoted by $\mathcal C_1$, and is generated by the Hadamard gate $H$ and the phase gate $S$. The full $n$-qubit Clifford group is generated by single-qubit Clifford gates together with CNOT gates.

\paragraph{Quantum Error-Correcting Codes.} We describe codes using the encoding operation, which is a unitary transformation on the input Hilbert space that we decompose as $A\otimes B \cong (\mathbb C^2)^{\otimes k} \otimes (\mathbb C^2)^{\otimes (n-k)}$. Here, the subsystem $A$ contains the $k$ logical input qubits, while the subsystem $B$ contains the $n-k$ ancilla qubits initialized to $|0\rangle^{\otimes (n-k)}$. Such a code is called an $[n,k]$ quantum error-correcting code. For an encoding unitary $U$, the associated code space is defined as the vector space $\{U(|\psi\rangle_A \otimes |0\rangle_B^{\otimes(n-k)}) \mid |\psi\rangle_A \in (\mathbb{C}^2)^{\otimes k}\}$. By considering a basis $\{|x\rangle \mid x\in \{0,1\}^k\}$ of $A$, we obtain a basis $\{|\bar{x}\rangle = U\bigl(|x\rangle\otimes|0\rangle^{\otimes(n-k)}\bigr) \mid x\in \{0,1\}^k\}$ of the code space. A code has distance at least $d+1$ if for all $x,y\in \{0,1\}^k, \mu \in \{0,1,2,3\}^n$ with $1\leq w(\mu)\leq d$, we have $\langle \bar{x} | \sigma_\mu | \bar y\rangle = C_\mu \delta_{xy}$ for some real numbers $C_\mu$ depending only on $\mu$ and not on $x,y$. A code with minimum distance $2d+1$ can correct $d$ errors. When $U\in\mathcal C_n$, this is a stabilizer code, and the distance of the code can be characterized using \cite[Proposition II.1]{Brown_2013}, whose proof we give in Appendix~\ref{sec:appendix-lemmas}:

\begin{customprop}{\ref{"BF_II1_proof"}}\label{prop: a1}
    A unitary $U\in\cC_n$ defines a quantum error-correcting code of distance at least $d+1$ if and only if for all $\nu_A \in \{0,1,2,3\}^k - \{0^k\}, \nu_B \in \{0,3\}^{n-k}$, and $\mu\in\{0,1,2,3\}^n$ of weight $1\leq w(\mu)\leq d$, we have $\tr[\sigma_\mu U (\sigma_{\nu_A} \otimes \sigma_{\nu_B})U^\dagger]\!=\!0$.
\end{customprop}

\subsection{Our Results}

Our main result is a proof that shallow random all-to-all circuits with a restricted gate support become good quantum error correction codes in depth $O(\log n)$. In Definitions~\ref{def:universality-conditions} and~\ref{definition: random matching layer} of Subsection~\ref{sec:architecture} we define a random circuit model, called a \textit{random matching circuit}, that takes as input a distribution $p_2$ on two-qubit Clifford gates and outputs a random circuit of depth $T$ using $nT/2$ independently random gates drawn from $p_2$. Given a distribution $p_2$ on $\cC_2$, a random matching circuit is constructed as follows: in each layer of the circuit, the $n$ qubits are paired up at random, and a random gate drawn from $p_2$ is applied to each pair of qubits. We count one parallel application of disjoint two-qubit gates as a matching round. A depth-$T$ random matching circuit is a product of $T$ independent matching rounds.

\begin{theorem}
    [Informal] Suppose $p_2$ is a distribution on the two-qubit Clifford group $\mathcal C_2$ satisfying some regularity conditions. For every fixed $\varepsilon>0$, $m>0$, and every
$\rho\in(0,3/4)$ satisfying $0\leq R \le 1-H(\rho)-\rho\log_2 3-\varepsilon$,
there is a constant $C=C(\varepsilon,m,p_2)$ such that a depth $C\log n$ random matching circuit drawn from $p_2$ defines an
$[n,\lfloor Rn\rfloor]$ stabilizer code of distance at least $\lfloor \rho n\rfloor+1$ with probability at least $1-n^{-m}-2^{-\Omega_\varepsilon(n)}$, for sufficiently large even $n$.
\end{theorem}

This result matches the tradeoff achieved by Brown and Fawzi in \cite{Brown_2013} with only $O(\log n)$ depth compared to their $O(\log^3 n)$. More precisely, we prove the following universality theorem for the ensemble of $n$-qubit random circuits we consider:

\begin{theorem}[see Theorem~\ref{thm:main-formal} for the full statement]
\label{theorem: main}
Suppose $p_2$ is a distribution on the two-qubit Clifford group $\cC_2$ satisfying three regularity conditions (see Definition~\ref{def:universality-conditions}). Let $m, \delta > 0$. Then, for sufficiently large even $n$, there exists a constant $c > 0$ depending on $\delta$, $m$, and $p_2$ such that a random matching circuit (Definition~\ref{definition: random matching layer}) of depth $\lceil c\log n\rceil $ derived from $p_2$ defines an $[n, k]$ quantum error-correcting code with distance at least $d + 1$ with probability at least \[
p_{\mathrm{code}} \geq 1 - n^{-m} - 2^{k - n(1 - H(d/n) - \log_2(3)d/n - \delta)}.
\]
\end{theorem}

For example, one can take $p_2$ to be the uniform distribution on $\cC_2$. In fact, the regularity conditions we impose are much weaker. A typical example of a distribution satisfying the three conditions needed for Theorem~\ref{theorem: main} to apply is given by twirling a CNOT gate by independent uniformly random one-qubit Cliffords (see Definition~\ref{definition: clifford-twirled cnot layer}). This gives rise to the following consequence: 

\begin{corollary}\label{cor:CNOTs-main}
Let $m, \delta > 0$. Then there is a random matching circuit (Definition~\ref{definition: clifford-twirled cnot layer}) of depth $O(\log n)$ using $O(n\log n)$ CNOT gates and $O(n\log n)$ random one-qubit Clifford gates that defines an $[n, k]$ quantum error-correcting code with distance at least $d + 1$ with probability at least \[
p_{\mathrm{code}} \geq 1 - n^{-m} - 2^{k - n(1 - H(d/n) - \log_2(3)d/n - \delta)}.
\]
\end{corollary}
The random one-qubit Clifford gates may be taken from a very small restricted gate set, such as $\{I, HS, (HS)^2\}$ (where $H$ is the Hadamard gate and $S$ is the phase gate); see Remark~\ref{rmk:cyclic-local-twirl}.

In particular, there exists a quantum circuit of depth $O(\log n)$ using $O(n\log n)$ CNOT gates and $O(n\log n)$ random one-qubit Clifford gates that asymptotically achieves the quantum Gilbert-Varshamov bound. Moreover, in the locally twirled construction, CNOT may be replaced by any fixed entangling two-qubit Clifford gate, with the gate and its inverse mixed symmetrically before applying the same independent twirls. This allows us to construct random circuits that produce good codes with high probability using a variety of extremely restricted gate sets.

Using a light cone argument, one can see that depth $O(\log n)$ is optimal for any circuit built from one- and two-qubit gates hoping to achieve linear code distance in $n$. In the next subsection in Proposition~\ref{prop: light-cone}, we will show that $O(n\log n)$ gates is also optimal within the random matching model (even those allowing fewer than $n/2$ gates per layer) hoping to achieve linear code distance with asymptotically positive probability, so our results are optimal within the independent random-matching model.

Finally, we provide a table that summarizes the results of the most relevant works in Table \ref{table: other encoder results}.
\begin{table}[h]
\centering
\begin{tabular}{c |c | c | c | c}
     Work & Encoder Ensemble & Log Depth & GV Rate-Distance & Additional Ancillas \\ 
     \hline
     \cite{Brown_2013} & Random pairs; uniform $\mathcal C_2$ & No & Yes & None \\ 
     \cite{cleve2016nearlinearconstructionsexactunitary} & Exact unitary 2-design & Yes & Yes (implied) & $\widetilde O(n)$ \\
     \cite{wills2026lineartimeencodabledecodablequantum} & Lossless-expander; CNOT only &Yes & Not established & None \\
     \textbf{Ours} & Independent matchings; twirled CNOT & Yes & Yes & None
\end{tabular}
\caption{Closest all-to-all exact-encoder results}
\label{table: other encoder results}
\end{table}

\subsection{Further Discussion}\label{sec:further-disc}

\paragraph{Optimality in the random matching model.} 
In any circuit architecture consisting only of one- and two-qubit gates, a qubit's forward light cone can grow by at most a multiplicative factor per layer. Producing an $n$-qubit code of linear distance $\delta n$ for any constant $\delta > 0$ requires each encoding qubit's light cone to reach at least $\delta n$ other qubits by the end of the circuit. Therefore, any such circuit of linear distance must have depth $\Omega(\log n)$ and use $\Omega(n)$ two-qubit gates. Our random matching construction achieves this lower bound on depth, but uses $O(n\log n)$ two-qubit gates, missing this na\"ive lower bound by a logarithmic factor. This answers a question of Brown--Fawzi \cite[Section IV]{Brown_2013} about achievability of the depth lower bound.

We can at least show (in Appendix~\ref{sec:appendix-lemmas}) that no ensemble in the independent random-matching model of Proposition~\ref{prop: light-cone} can achieve an $O(n)$ gate count with constant success probability, even if it is made sparser by only matching a subset of the $n$ qubits at each layer.
\begin{customprop}{\ref{prop: light-cone}}
    Consider an ensemble of $T$ layers in which layer $r$ places arbitrary two-qubit gates on a uniformly random set of fixed size $B_r \leq  n/2$ disjoint pairs, independently of earlier layers. Let $N\coloneqq \sum_{r=1}^T B_r$ be the total number of two-qubit gates. Suppose that, with probability at least $p$, the resulting Clifford encoder defines an $[n,k]$ stabilizer code of distance at least $d+1$, where $k\geq 1$. Then
    \[
        N\geq \frac{n-1}{2}\ln\bigl(p(d+1)\bigr).
    \]
    Therefore, if $p$ is bounded below by a positive constant and $d\geq\delta n$ for some constant $\delta>0$, then $N=\Omega(n\log n)$ and $T=\Omega(\log n)$.
\end{customprop} 
Thus, a construction that achieves the $\Omega(n)$-gate lower bound would likely need to be highly structured or at least allow for dependence between random layers.

\paragraph{Restricted gate set.} Brown and Fawzi sample each two-qubit gate uniformly from the full two-qubit Clifford group $\mathcal C_2$. Although such gates can be compiled into a constant number of elementary Clifford gates, our principal construction uses CNOT as the only entangling gate, with the remaining randomness supplied by independent one-qubit Clifford twirls. Theorem~\ref{theorem: main} constructs a quantum error correction code using a random circuit with a more restricted gate set, namely standard locally-twirled two-qubit Clifford gates, and Corollary \ref{cor:CNOTs-main} specializes this to the case where CNOT is the only entangling gate, addressing the restricted-gate question raised in \cite{Brown_2013}. Moreover, our circuit primitive also has a direct hardware motivation: recent trapped-ion experiments implemented $98$-qubit random Clifford layers using uniformly random pairings, random one-qubit Cliffords, and a fixed maximally entangling $R_{ZZ}(\pi/2)$ gate on every pair \cite{quantinuum2025helios}. This experiment closely parallels our ensemble, with the key differences being that our bounds remain logical all-to-all circuit bounds, which do not include device-specific routing or transport time, restrictions on simultaneous gates, encoding noise, fault-tolerant state preparation, or decoding.

\subsection{Limitations and Open Problems.}

The regularity conditions on $p_2$ under which we prove Theorem~\ref{theorem: main} are (see Definition~\ref{def:universality-conditions}): \begin{enumerate}
    \item Bi-invariance: there exists a subgroup $\mathcal{H}$ of the one-qubit projective Clifford group $\cC_1$ such that $\mathcal{H}$ acts transitively on $\{X, Y, Z\}$ up to sign and $p_2$ is invariant under left and right multiplication by $\mathcal{H} \otimes \mathcal{H}$.
    \item Positive Entanglement: $p_2$ assigns a positive probability to an entangling gate;
    \item Reversibility: For $U \sim p_2$, we have $U^\dagger \sim p_2$.
\end{enumerate}
Clearly, positive entanglement is a necessary condition for a random matching circuit derived from $p_2$ to produce a good code with any nonzero probability. However, while the assumptions of bi-invariance and reversibility are crucial to the structure of our argument for technical reasons, they are not obviously necessary for a random matching circuit derived from $p_2$ to produce a good quantum error-correcting code (although \textit{some} assumption is still necessary to replace bi-invariance). For example, our current bi-invariance assumption allows twirling by a random element from a three-element gate set $\{I, HS, (HS)^2\}$, but not from some other natural three-element gate sets such as $\{I, H, S\}$. It would be interesting to see if these conditions could be relaxed.

Next, we give an existence result for the quantum code, rather than a deterministic explicit family. Therefore, successfully derandomizing the matching sequence and the twirling process would be very promising directions. Similarly, our work only studies the encoding complexity. While our construction gives a high-distance stabilizer code, it does not yield an efficient decoder. Hence, designing an efficient decoder for our ensemble is an open problem.  

Within the independent random-matching model, Proposition~\ref{prop: light-cone} shows that an ensemble achieving distance $d+1$ with probability at least $p$ must use $N \geq \frac{n-1}{2}\ln\bigl(p(d+1)\bigr)$ two-qubit gates. Hence, our gate count is therefore optimal within this model.

\subsection{Related Work}
\label{sec: related-work}

The closest predecessors to our work are the scrambling, decoupling, and coding results of Brown and Fawzi \cite{brown2013scramblingspeedrandomquantum, brownDecouplingRandomQuantum2015, Brown_2013}. Their scrambling work already studied parallel random-matching circuits on the complete graph and proved $O(\log n)$-depth scrambling for a constant-size message. Their stronger exact-code result used $O(n\log^2 n)$ sequential random Clifford gates and, after parallelization, yielded depth $O(\log^3 n)$.  They showed that for any $\delta > 0$ random Clifford circuits of size $O(n \log^2 n)$ can encode $k$ logical qubits into $n$ physical qubits with distance $d+1$ whenever $\frac{k}n < 1 - H(\frac{d}{n}) - \frac dn \log_2 3 - \delta$, which matches the quantum Gilbert-Varshamov tradeoff \cite{gottesman1997stabilizer} up to an arbitrarily small fixed slack. Wills et al.~\cite{wills2026lineartimeencodabledecodablequantum} construct randomized and explicit asymptotically good CSS codes at every prescribed constant rate. Their CNOT-only encoding and unencoding circuits have $O(\log n)$ depth and $O(n)$ gates, and they provide efficient classical decoding algorithms. Although their construction is advantageous in gate count, explicitness, and decoding, their exact rate--distance tradeoff is explicitly left open. 

Our contribution is to obtain the quantum-GV rate-distance tradeoff in $O(\log n)$ random-matching layers, while also allowing restricted locally twirled two-qubit Clifford distributions such as the Clifford-twirled CNOT.

\paragraph{Local and brickwork random-circuit codes.} A parallel line of recent work studies random-circuit codes under geometric locality constraints. \cite{PhysRevX.11.031066,PhysRevResearch.7.013040} investigated quantum codes generated by low-depth random circuits with local connectivity in spatial dimension $D$, showing that local random circuits can already produce strong coding behavior, especially for erasure noise. \cite{Darmawan_2024} studied one-dimensional logarithmic-depth random Clifford encoders against Pauli noise using tensor-network maximum-likelihood decoding, giving evidence that such local random encoders can approach hashing-bound behavior despite their geometric constraints. More recently, \cite{Liu_2026} proved approximate quantum error-correction guarantees for one-dimensional logarithmic-depth random Clifford circuits, and \cite{Kroll:2026max} proved error-correction results for one-dimensional brickwork Clifford circuits, including logarithmic-depth approximate correction and matching bounds for exact correction in their model.

These local and brickwork results are complementary to ours. They impose much stronger geometric constraints on the interaction graph, often in one spatial dimension, and obtain approximate or channel-specific error-correction guarantees. In contrast, our general theorem applies to a class of locally Clifford-invariant two-qubit gate distributions, and our principal restricted-gate construction uses CNOT as its only entangling operation. Our goal is therefore different: we ask how little two-qubit gate randomness is needed to recover \cite{Brown_2013}'s exact-distance guarantee for stabilizer codes. In this sense, the brickwork literature moves toward spatial locality, while our work addresses gate-set restriction.

\paragraph{Unitary $t$-designs.} Random quantum circuits are frequently studied as efficient approximations to Haar randomness \cite{haferkamp2022random,Dankert_2009,Brand_o_2016, q172-8cmt, cleve2016nearlinearconstructionsexactunitary}. A unitary $t$-design is an ensemble whose $t$-fold moment operator agrees with the Haar $t$-fold moment operator, and approximate designs have applications throughout quantum information, including randomized benchmarking \cite{zhang2017randomizedbenchmarkingusingunitary}, decoupling and quantum cryptography. Unitary designs provide a natural benchmark for shallow random encoders. The uniform Clifford group is an exact unitary 2-design \cite{Dankert_2009}, and random-circuit constructions of approximate designs were developed in \cite{Harrow_2009,Brand_o_2016}, with improved depth bounds in \cite{haferkamp2022random,schuster2025random}. \cite{Harrow_2009} showed that polynomial-size random quantum circuits form approximate unitary $2$-designs, and \cite{Brand_o_2016} proved that local random circuits form approximate unitary $t$-designs for general $t$. 
Our result is related to this literature through the use of second moments; however, the code-distance argument only requires control of a specific part of the second moment: the probability that a nontrivial Pauli operator evolves to a low-weight Pauli operator.

Most related to our work is \cite{cleve2016nearlinearconstructionsexactunitary} which constructed near-linear-size Pauli-mixing Clifford ensembles that form exact unitary 2-designs. Their construction also implies the quantum Gilbert–Varshamov distance tradeoff with exponentially high probability.  Their unconditional Clifford-based implementation has $O(\log^2 n)$ depth and uses $\widetilde O(n)$ additional work ancillas; conversely, our result instead gives an unconditional encoder of depth $O(\log n)$ and size $O(n\log n)$ without additional work ancillas, generated by independent random-matching layers with a fixed entangling gate and local twirls. Rather than establishing a full unitary design, we prove our result by establishing the low Pauli-weight estimate needed in the union bound.

\paragraph{Explicit and LDPC quantum codes.} Another major direction seeks explicit families of quantum codes with sparse parity checks. Topological codes, including toric and surface codes \cite{Kitaev_2003,Dennis_2002,Fowler_2012}, have local stabilizer checks and strong practical appeal, but geometric locality imposes rate-distance tradeoffs: for instance, two-dimensional local stabilizer codes must obey the Bravyi-Poulin-Terhal bound \cite{PhysRevLett.104.050503}. More recently, hypergraph product codes of \cite{10.1109/TIT.2013.2292061} gave quantum LDPC codes with positive rate and distance proportional to the square root of the block-length. This initiated a sequence of breakthroughs improving the asymptotic parameters of quantum LDPC codes, including fiber-bundle codes \cite{Hastings_2021} and balanced-product codes \cite{Breuckmann_2021} which broke the earlier square-root distance barrier.

The recent resolution of the quantum LDPC conjecture in \cite{panteleev2022asymptoticallygoodquantumlocally} produced asymptotically good quantum LDPC codes using lifted products over non-abelian groups. Quantum Tanner codes of \cite{leverrier2022quantumtannercodes} provide a related expander-based construction with good rate and linear distance. A further line of work has developed efficient decoders for these good quantum LDPC codes \cite{gu2022efficientdecoderlineardistance,leverrier2022efficientdecodingconstantfraction}, including linear-time decoding results \cite{dinur2022goodquantumldpccodes} and single-shot decoding guarantees for quantum Tanner codes \cite{Gu_2024}. In contrast,  we do not address efficient decoding. Instead, we focus on the complexity and physical simplicity of the encoder, proving that short random circuits built from CNOT gates and one-qubit Clifford gates already suffice to obtain high-distance stabilizer codes that achieve the quantum GV tradeoff.

\section{Circuit Architecture and Proofs of Main Results}

In this section we will describe a random circuit architecture and show that it produces a good quantum error-correcting code. The circuit architecture takes as input a distribution on two-qubit Clifford gates, which we allow to be any distribution satisfying three reasonable conditions (Definition~\ref{def:universality-conditions}). We describe this architecture in Subsection~\ref{sec:architecture}.

The remaining subsections are dedicated to proving that this architecture produces a good quantum error-correcting code assuming the results of Section~\ref{sec:markov-chain} as a black box. The method is similar to that of Brown--Fawzi \cite{Brown_2013}: in Subsection~\ref{sec:second-moment} we define a \textit{second moment operator} for the random Clifford produced by our circuit and in Subsection~\ref{sec:circuit-to-code} we relate the behavior of this second moment operator to the probability of producing a code of given distance. Since our circuit architecture consists of a sequence of random layers, the second moment operator of the circuit defines a Markov chain on length-$n$ strings of Pauli matrices. We study this Markov chain in Subsection~\ref{sec:second-moment} and show that the three conditions of Definition~\ref{def:universality-conditions} imply some properties of the Markov chain that we will use in Section~\ref{sec:markov-chain} to give good bounds on the quantities we need to control to produce a good code with high probability.

\subsection{Circuit Architecture}\label{sec:architecture}
Here, we describe our circuit architecture. Fix even $n$. We consider an all-to-all random matching architecture on $n$ physical qubits, where each circuit layer is a Clifford-twirled CNOT matching layer, as given in Definition~\ref{definition: clifford-twirled cnot layer}.

\begin{definition}\label{def:universality-conditions}
Let $p_2$ be a distribution on the two-qubit Clifford group $\cC_2$. We define the following three properties: \begin{enumerate}
    \item Bi-invariance: there exists a subgroup $\mathcal{H}$ of the one-qubit projective Clifford group $\cC_1$ such that $\mathcal{H}$ acts transitively on $\{X, Y, Z\}$ up to sign and $p_2$ is invariant under left and right multiplication by $\mathcal{H} \otimes \mathcal{H}$.
    \item Positive Entanglement: $p_2$ assigns a positive probability to an entangling gate;
    \item Reversibility: For $U \sim p_2$, we have $U^\dagger \sim p_2$.
\end{enumerate}
Here, acting transitively up to sign means that $\mathcal{H}$ acts transitively on the set $\{X \otimes X, Y \otimes Y, Z \otimes Z\}$.
\end{definition}\looseness=-1

An example of a subgroup $\mathcal{H}_1\subseteq \cC_1$ satisfying the condition needed to witness bi-invariance is $\cC_1$ itself. For example, the uniform distribution on $\cC_2$ satisfied bi-invariance with this choice of subgroups. For a more interesting example illustrating the condition that $\mathcal{H}$ acts on $\{X, Y, Z\}$ transitively up to sign, see Remark~\ref{rmk:cyclic-local-twirl}.

\begin{definition}[Random matching layer $\Gamma$]\label{definition: random matching layer}
Let $p_2$ be a distribution on the two-qubit Clifford group $\cC_2$. A random $p_2$-matching layer $\Gamma$ is sampled as follows.
First sample a uniformly random ordered perfect matching $\mathcal{O}$ of $[n]$ which partitions the qubits into $n/2$ disjoint ordered pairs. For every ordered pair $(i,j)\in \mathcal{O}$, let $G_{ij}$ be a random two-qubit Clifford gate sampled according to $p_2$, applied to the $i$th and $j$th qubits. Then, the layer $\Gamma = \prod_{(i,j)\in O} G_{ij}$ is given by applying all gates $G_{ij}$ for $(i,j)\in \mathcal{O}$ in parallel.
\end{definition}

Then, our depth-$T$ encoder is the product of $T$ independent layers (see Figure~\ref{fig:random-matching-architecture}), which is given by the unitary \begin{equation}U_T = \Gamma^{(T)} \Gamma^{(T-1)} \cdots \Gamma^{(1)}.\label{equation: encoder}\end{equation} 

\begin{definition}[Clifford-twirled CNOT layer $\Gamma_{\mathrm{CNOT}}$]\label{definition: clifford-twirled cnot layer}
We define $G \sim p_{2, \mathrm{CNOT}}$ by $G = (A\otimes B) \mathrm{CNOT}_{1\to2} (C\otimes D)$, where $A, B, C, D \sim \mathrm{Unif}(\cC_1)$ are independent uniformly random single-qubit Clifford gates. 

Then let $\Gamma_{\mathrm{CNOT}}$ be the random matching layer (Definition~\ref{definition: random matching layer}) obtained from $p_{2, \mathrm{CNOT}}$. We denote by $U^{\mathrm{CNOT}}_T$ the associated depth-$T$ circuit from \eqref{equation: encoder}.
\end{definition}
Note that since $\mathrm{CNOT}$ is entangling and its own inverse, the distribution $p_{2, \mathrm{CNOT}}$ satisfies all three conditions of Definition~\ref{def:universality-conditions}.

\begin{remark} [Cyclic local twirls] \label{rmk:cyclic-local-twirl}
    A three-element cyclic Clifford twirl is sufficient for the Clifford-twirled CNOT construction. Similar cyclic Pauli uniformization twirls are used in the construction of approximate unitary designs in \cite{Dankert_2009}.  Let $R\coloneqq HS$, where $H$ is the Hadamard gate and $S$ is the phase gate. Conjugation by $R$ cyclically permutes the non-identity Pauli operators up to sign: \[ RXR^\dagger=-Y,\qquad RYR^\dagger=-Z,\qquad RZR^\dagger=X. \] 
    Therefore, the subgroup of $\cC_1$ generated by $R$ indeed acts transitively on $\{X, Y, Z\}$ up to sign. 

    Note that $R$ has order 3 in the projective Clifford group $\cC_1/U(1)$ because $R^3 = e^{i\pi/4}I$. Thus the uniform distribution on $\{I, R, R^2\}$ is invariant under multiplication by $R$ up to global phase. Consequently, one can replace $A, B, C, D$ in Definition~\ref{definition: clifford-twirled cnot layer} with random gates from $\{I, R, R^2\}$ without changing the conclusion of Corollary~\ref{cor:CNOTs-main}.
\end{remark}

Each layer of $U_T^{\mathrm{CNOT}}$ contains exactly $n/2$ CNOT gates arranged on a uniformly random ordered perfect matching, together with local random single-qubit Clifford gates before and after the CNOT layer (see Figure~\ref{fig:twirled-cnot-architecture}). So, for a depth $T$ circuit, the total number of CNOT gates is $nT/2$. We use $U_T$ as the encoding Clifford for an $[n,k]$ stabilizer code by applying it to $k$ logical input qubits and $n-k$ ancilla qubits initialized to $|0\rangle$.

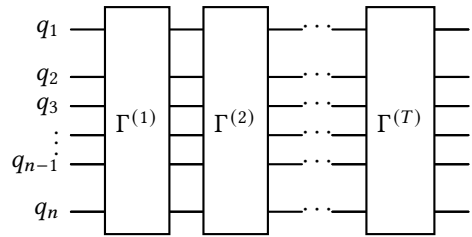
\begin{figure}[hbt!]
\centering
\[
\begin{quantikz}[row sep=0.28cm, column sep=0.45cm]
\lstick{$q_1$}
    & \gate[6]{\Gamma^{(1)}} 
    & \gate[6]{\Gamma^{(2)}} 
    & \cdots
    & \gate[6]{\Gamma^{(T)}} 
    & \qw \\
\lstick{$q_2$}
    & \qw & \qw & \cdots & \qw & \qw \\
\lstick{$q_3$}
    & \qw & \qw & \cdots & \qw & \qw \\
\lstick{$\vdots$}
    & \qw & \qw & \cdots & \qw & \qw \\
\lstick{$q_{n-1}$}
    & \qw & \qw & \cdots & \qw & \qw \\
\lstick{$q_n$}
    & \qw & \qw & \cdots & \qw & \qw
\end{quantikz}
\]
\caption{The encoder $U_T=\prod_{t=0}^{T-1} \Gamma^{(T-t)}$ is a product of random matching layers, as in Definition~\ref{definition: random matching layer}.}
\label{fig:random-matching-architecture}
\end{figure}

\begin{figure}[H]
\centering
\begin{quantikz}[row sep=0.28cm, column sep=0.23cm]
\lstick{$q_1$}
    & \gate{L^{(1)}_1}
    & \ctrl{1}
    & \gate{R^{(1)}_1}
    & \qw
    & \gate{L^{(2)}_1}
    & \ctrl{5}
    & \qw
    & \qw
    & \gate{R^{(2)}_1}
    & \qw \\
\lstick{$q_2$}
    & \gate{L^{(1)}_2}
    & \targ{}
    & \gate{R^{(1)}_2}
    & \qw
    & \gate{L^{(2)}_2}
    & \qw
    & \ctrl{3}
    & \qw
    & \gate{R^{(2)}_2}
    & \qw \\
\lstick{$q_3$}
    & \gate{L^{(1)}_3}
    & \ctrl{1}
    & \gate{R^{(1)}_3}
    & \qw
    & \gate{L^{(2)}_3}
    & \qw
    & \qw
    & \ctrl{1}
    & \gate{R^{(2)}_3}
    & \qw \\
\lstick{$q_4$}
    & \gate{L^{(1)}_4}
    & \targ{}
    & \gate{R^{(1)}_4}
    & \qw
    & \gate{L^{(2)}_4}
    & \qw
    & \qw
    & \targ{}
    & \gate{R^{(2)}_4}
    & \qw \\
\lstick{$q_5$}
    & \gate{L^{(1)}_5}
    & \ctrl{1}
    & \gate{R^{(1)}_5}
    & \qw
    & \gate{L^{(2)}_5}
    & \qw
    & \targ{}
    & \qw
    & \gate{R^{(2)}_5}
    & \qw \\
\lstick{$q_6$}
    & \gate{L^{(1)}_6}
    & \targ{}
    & \gate{R^{(1)}_6}
    & \qw
    & \gate{L^{(2)}_6}
    & \targ{}
    & \qw
    & \qw
    & \gate{R^{(2)}_6}
    & \qw
\end{quantikz}
\caption{
Two Clifford-twirled CNOT matching layers. 
Each $L_i^{(t)},R_i^{(t)}\in\mathcal C_1$ is an independent one-qubit Clifford. 
Layer $\Gamma^{(1)}_{\mathrm{CNOT}}$ uses $\mathcal{O}=\{(1,2),(3,4),(5,6)\}$, while layer $\Gamma^{(2)}_{\mathrm{CNOT}}$ uses $\mathcal{O}=\{(1,6),(2,5),(3,4)\}$.  
}
\label{fig:twirled-cnot-architecture}
\end{figure}
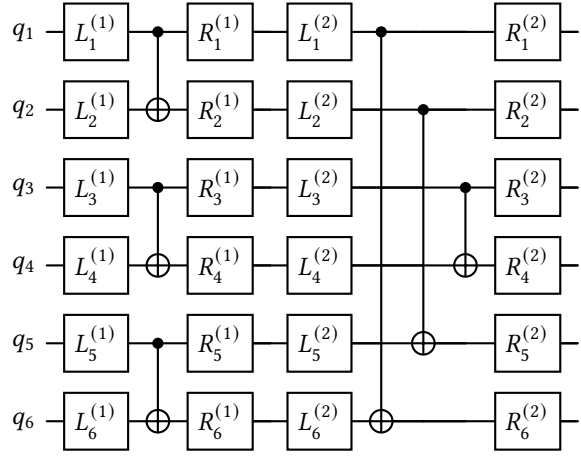 

\subsection{Second Moment Operators}\label{sec:second-moment}

Consider a random circuit with the above architecture. Applying all the gates in this random circuit to $n$ qubits is equivalent to applying some random unitary $U_T \in \cC_n$. This random circuit thus defines a measure $\mu_{\mathrm{circ}}$ over unitary transformations on $n$ qubits. Any such measure defines a \textit{second moment operator}:

\begin{definition}[Second moment operator]\label{def:second-moment}
Let $B((\mathbb{C}^2)^{\otimes n})$ be the space of operators acting on the Hilbert space associated to $n$ qubits. Let $p$ be a probability measure on $\cC_n$.

The \textit{second moment operator} $M_p$ of $p$ is a superoperator acting on $B((\mathbb{C}^2)^{\otimes n}) \otimes B((\mathbb{C}^2)^{\otimes n})$ as follows: \[
M_p[X \otimes Y] = \E_{U \sim p}[(U X U^\dagger) \otimes (U Y U^\dagger)]
\]
for any operators $X, Y$ on $n$ qubits. We will sometimes also write $M_U$ for $M_p$ if $U \sim p$.
\end{definition}
Note that if $U, V$ are independent random operators on $(\C^2)^{\otimes n}$, then $M_{UV} = M_UM_V$.

For the purposes of defining good quantum error-correcting codes from random Clifford operators, we only care about the action of the second moment operator on pairs of identical Pauli strings $\sigma_\nu \otimes \sigma_\nu$, where $\nu \in \{0, 1, 2, 3\}^n$. (Here $\sigma_\nu = \sigma_{\nu_1} \otimes\dots \otimes \sigma_{\nu_n}$.) For the remainder of this subsection we will focus exclusively on this domain.

Let $p_2$ be a bi-invariant distribution on the two-qubit Clifford group $\cC_2$. We denote by $M_{\mathrm{circ}, p_2}$ the second moment operator of a single layer $\Gamma$ of the resulting random matching circuit. The second moment operator of a $T$-layer circuit is given by $M_{\mathrm{circ}, p_2}^T$. We will analyze its action on pairs of Pauli strings $\sigma_{\nu} \otimes \sigma_\nu$, where $\nu \in \{0, 1, 2, 3\}^n$. 

To start, note that by bi-invariance of $p_2$, the distribution of a random layer $\Gamma$ is invariant under left and right multiplication by $\mathcal{H} \otimes \mathcal{H}$ for some subgroup $\mathcal{H} \subseteq \cC_1$ acting transitively on $\{X, Y, Z\}$ up to sign. Let $\phi^{(n)}$ be the second moment operator for $n$ independent uniformly random gates drawn from $\mathcal{H}$ applied in parallel. We have \[
M_{\mathrm{circ}, p_2} = \phi^{(n)} \circ M_{\mathrm{circ}, p_2} \circ \phi^{(n)}.
\]
Let $\phi^{(1)}_i$ be the second moment operator for a uniformly random gate drawn from $\mathcal{H}$ applied to qubit $i$. We have $\phi^{(n)} = \prod_{i=1}^n \phi^{(1)}_i$. We note that  $\mathcal{H}$ acts transitively on $\{X \otimes X, Y \otimes Y, Z \otimes Z\}$ via $T\cdot(\sigma_\alpha\otimes\sigma_\alpha)\coloneqq T\sigma_\alpha T^\dagger \otimes T\sigma_\alpha T^\dagger$. Thus for $\alpha\in \{0, 1, 2, 3\}$ we have\begin{equation}\label{eq:one-qubit-twirl}
\phi^{(1)}[\sigma_\alpha \otimes \sigma_\alpha] = \begin{cases}
    \sigma_0 \otimes \sigma_0 &\text{if }\alpha = 0 \\
    \frac{1}{3}\sum_{\beta \in \{1, 2, 3\}} \sigma_\beta \otimes \sigma_\beta &\text{if }\alpha \neq 0.
\end{cases}
\end{equation}
Now let $\nu \in \{0, 1, 2, 3\}^n$. Let $w(\nu)$ be the number of nonzero entries in $\nu$. For a string $x \in \{0, 1\}^n$, let $w(x)$ be the number of nonzero entries in $x$; for a string $\nu \in \{0, 1, 2, 3\}^n$, let $b(\nu) \in \{0, 1\}^n$ be the string with zeroes in the same positions as in $\nu$. Then let\[
E_x \coloneqq \frac{1}{3^{w(x)}}\sum_{\substack{\upsilon \in \{0, 1, 2, 3\}^n \\ b(\upsilon) = x}} \sigma_\upsilon \otimes \sigma_\upsilon.
\]
We have \[
\frac{1}{4^n}\tr[(\sigma_\nu \otimes \sigma_\nu)E_x] = \begin{cases}
    0 & \text{if } b(\nu) \neq x \\
    3^{-w(x)}  &\text{if } b(\nu) = x,
\end{cases} \qquad \text{ and so }\qquad  \frac{1}{4^n}\tr[E_xE_y] = \begin{cases}
    0 & \text{if } x \neq y \\
    3^{-w(x)}  &\text{if } x = y.
\end{cases}
\]
In particular, the $E_x$ are linearly independent, and we write \[
\mathcal{E} \coloneqq \operatorname{span}\{E_x \mid x\in\{0, 1\}^n\}.
\] Iterating \eqref{eq:one-qubit-twirl} we get \[
\phi^{(n)}[\sigma_\nu \otimes \sigma_\nu] = E_{b(\nu)}
\]
In other words, the second moment operator for $n$ independent uniformly random elements of $\mathcal{H}$ sends any pair of identical Pauli strings to the average over all pairs which have nonidentity Paulis in the same locations. It follows that the action of the second moment operator $M_{\mathrm{circ}, p_2}$ on pairs of identical Pauli strings is determined by how it acts on the basis states $E_x$ for strings $x \in \{0, 1\}^n$, i.e., by its restriction to $\operatorname{span}\{E_x \mid x\in\{0, 1\}^n\}$. For any string $\nu \in \{0, 1, 2, 3\}^n$ we have $M_{\mathrm{circ}, p_2}[\sigma_\nu\otimes\sigma_\nu] = M_{\mathrm{circ}, p_2}[E_{b(\nu)}] \in \mathcal{E}$.

Following \cite[Section II.C]{Brown_2013} we can encode the action of $M_{\mathrm{circ}, p_2}$ on pairs of identical Pauli strings using a $4^n \times 4^n$ matrix \[
Q(\nu, \nu') = \frac{1}{4^n}\tr[(\sigma_{\nu'}\otimes \sigma_{\nu'})M_{\mathrm{circ}, p_2}(\sigma_\nu \otimes \sigma_\nu)].
\]
Since $M_{\mathrm{circ}, p_2}(\sigma_\nu \otimes \sigma_\nu)$ is the average of a random Clifford applied to $\sigma_\nu \otimes \sigma_\nu$, it can be written as a linear combination of Pauli strings with positive coefficients summing to 1. Thus, for any Pauli string $\nu$ we have $\sum_{\nu' \in \{0, 1, 2, 3\}^n}Q(\nu, \nu') = 1$, and we can interpret $Q$ as the transition matrix for a Markov chain on the state space on length $n$ Pauli strings. A Pauli string transitions according to this Markov chain simply by applying a random matching layer $\Gamma$. Later, we will see that the distribution of this Markov chain after $T$ steps governs the quality of the depth-$T$ encoding circuit described in Subsection~\ref{sec:architecture}. In fact, we can reduce to studying a simpler Markov chain.  Using the knowledge that $M_{\mathrm{circ}, p_2}(\sigma_\nu \otimes \sigma_\nu)$ only depends on $b(\nu)$ we can compress the information in $Q$ into a $2^n\times 2^n$ matrix \[
Q_0(y, y') = \frac{3^{w(y')}}{4^n}\tr[E_{y'}M_{\mathrm{circ}, p_2}(E_y)]
\]
with $Q_0(b(\nu), b(\nu')) = 3^{w(b(\nu'))}Q(\nu, \nu')$. We see that $Q_0$ is the transition matrix for the Markov chain describing the locations of nonidentity Paulis in strings evolving according to $Q$. Indeed, for any binary strings $y$, $y'$ and any Pauli string $\nu$ with $b(\nu) = y$ we have \[
Q_0(y, y') = \sum_{\substack{\nu' \in \{0, 1, 2, 3\}^n \\ b(\nu') = y'}} Q(\nu, \nu')
\]
The remainder of this subsection is dedicated to describing some nice properties of the Markov chain $Q_0$ which are consequences of its construction. We will use these properties in Section~\ref{sec:markov-chain} to have sufficient control over the dynamics of $Q$ to ensure that our depth-$T$ encoding circuit produces good codes when $T$ is large enough.

We have \[
M_{\mathrm{circ}, p_2} = \mathbb E_\mathcal{O}\left[\prod_{(i, j) \in \mathcal O} m_{ij}\right]
\]
where $\mathcal{O}$ is a uniformly random ordered perfect matching on $[n]$ and $m_{ij}$ is the second moment operator for $p_2$ acting on only the $i$th and $j$th qubits. By the same argument as above, the second moment operator $M_{p_2}$ associated to $p_2$ is determined by how it acts on the four basis states $E_{00}$, $E_{01}$, $E_{10}$, and $E_{11}$. This operator also gives rise to a transition matrix on binary strings of length 2: \[
K_{p_2} = \begin{pmatrix}
    1 & 0 & 0 & 0 \\
    0 & p_{01\to01} & p_{01\to10} & p_{01\to11} \\
    0 & p_{10\to01} & p_{10\to10} & p_{10\to11} \\
    0 & p_{11\to01} & p_{11\to10} & p_{11\to11}
\end{pmatrix}
\]
The string $00$ is isolated because it is impossible to obtain the identity matrix by conjugating a nonidentity Pauli string by any Clifford operator. As such, this matrix is determined by the bottom-right $3\times 3$ corner, which we call $K'_{p_2}$.

One step of the Markov chain $Q_0$ is performed on a string $x$ by choosing a uniformly random ordered perfect matching $\mathcal{O}$ of $[n]$. Then, for each $(i, j) \in \mathcal{O}$, we evolve $x_ix_j$ independently according to $K_{p_2}$. Markov chains of this form are studied in some generality in Section~\ref{sec:markov-chain}. In Lemma~\ref{lem:stationary}, we determine properties of the Markov chain $Q_0$ from properties of $K'_{p_2}$ that we prove in the following lemma:

\begin{lemma}\label{lem:K'-reversible}
Let $p_2$ be a distribution on the two-qubit Clifford group $\cC_2$ satisfying bi-invariance, positive entanglement, and reversibility (Definition~\ref{def:universality-conditions}).
Let \begin{equation}\label{eq:K'}
K'_{p_2} = \begin{pmatrix}
p_{01\to01} & p_{01\to10} & p_{01\to11} \\
p_{10\to01} & p_{10\to10} & p_{10\to11} \\
p_{11\to01} & p_{11\to10} & p_{11\to11}
\end{pmatrix}
\end{equation} be the $3\times 3$ row-stochastic matrix constructed from the second moment operator associated to $p_2$ as above. Then $K'_{p_2}$ is reversible with respect to the stationary distribution $\pi' = (1/5, 1/5, 3/5)$ on $01, 10, 11$, and $1/3 \leq p_{11\to11} < 1$.
\end{lemma}
\begin{proof}
Let $U$ be a random two-qubit Clifford gate drawn from $p_2$. The matrix $K'_{p_2}$ is the transition matrix describing the positions of the nonidentity Pauli operators after applying $U$ to a Pauli string of length 2.

We check the detailed balance condition for $K'_{p_2}$ with respect to $\pi'$. Let $x$ and $y$ be two nonzero bit strings of length 2. The probability $K'_{p_2}(x, y)$ of moving from $x$ to $y$ in one step of $K'_{p_2}$ is given by \[
3^{w(x)}K'_{p_2}(x, y) = \sum_{\substack{\nu \in \{0, 1, 2, 3\}^2 - \{00\} \\ b(\nu) = x}} \PP[b(U\sigma_\nu U^\dagger) = y].
\]
Now by bi-invariance, any two Pauli strings with $b(U\sigma_{\nu} U^\dagger) = y$ are equally likely to appear as $U\sigma_{\nu} U^\dagger$. Therefore we can write \[
\sum_{\substack{\nu \in \{0, 1, 2, 3\}^2 - \{00\} \\ b(\nu) = x}} \PP[b(U\sigma_\nu U^\dagger) = y] = \sum_{\substack{\nu \in \{0, 1, 2, 3\}^2 - \{00\} \\ b(\nu) = x}}\sum_{\substack{\upsilon \in \{0, 1, 2, 3\}^2 - \{00\} \\ b(\upsilon) = y}} \PP[U\sigma_{\nu} U^\dagger = \pm\sigma_{\upsilon}]
\]
By reversibility of $p_2$ we have \[
\PP[U\sigma_{\nu} U^\dagger = \pm\sigma_{\upsilon}] =  \PP[\pm\sigma_{\nu} = U\sigma_{\upsilon}U^\dagger]
\]
and we obtain \[
3^{w(x)}K'_{p_2}(x, y) = \sum_{\substack{\nu \in \{0, 1, 2, 3\}^2 - \{00\} \\ b(\nu) = x}}\sum_{\substack{\upsilon \in \{0, 1, 2, 3\}^2 - \{00\} \\ b(\upsilon) = y}} \PP[U\sigma_{\upsilon} U^\dagger = \pm\sigma_{\nu}] = 3^{w(y)}K'_{p_2}(y, x)
\]
Since $\pi'(x)$ is proportional to $3^{w(x)}$, this shows that $K'_{p_2}$ is reversible with respect to $\pi'$; in particular, $\pi'$ is stationary for $K'_{p_2}$.

Since $\pi'$ is stationary for $K'_{p_2}$, we have $p_{10\to11} + p_{01\to11} = 3 - 3p_{11\to11}$. Since the left hand side is at most $2$, we have $p_{11\to11} \geq 1/3$. To show that $p_{11\to11} < 1$, we will show that $p_{11\to01} > 0$. Indeed, by reversibility , it is equivalent to show that $p_{01\to11} > 0$. To do this, we will show that any entangling gate sends some string of the form $\sigma_0 \otimes \sigma_i$ to a Pauli string of weight 2, where $i \in \{1, 2, 3\}$.

Let $U_0$ be a two-qubit Clifford gate in the support of $p_2$. For $i \in \{1, 2, 3\}$, suppose $U_0(\sigma_0\otimes \sigma_i)U_0^\dagger = \omega_i(\sigma_{\alpha_i} \otimes \sigma_{\beta_i})$, where $\alpha_i, \beta_i \in \{0, 1, 2, 3\}$ (but not $\alpha_i = \beta_i = 0$) and $\omega_i$ is a phase. Suppose that for each $i$, either $\alpha_i = 0$ or $\beta_i = 0$. We want to show that $U_0$ is not entangling.

We start by noticing that either $\alpha_i = 0$ for all $i$ or $\beta_i = 0$ for all $i$. Indeed, suppose $\alpha_i = 0$ and $\alpha_j \neq 0$ (hence $\beta_j = 0$) for some $i \neq j$. Then $U_0(\sigma_0\otimes \sigma_i)U_0^\dagger$ and $U_0(\sigma_0\otimes \sigma_j)U_0^\dagger$ commute, while $\sigma_0 \otimes \sigma_i$ and $\sigma_0\otimes\sigma_j$ do not, a contradiction.

We may assume $\alpha_i = 0$ for all $i$. Otherwise, replace $U_0$ by $\operatorname{SWAP}U_0$, where $\mathrm{SWAP}$ denotes the two-qubit Clifford gate swapping two qubits, which is not entangling. We have that $U_0$ is entangling if and only if $\operatorname{SWAP}U_0$ is.

Now since the Pauli operators span the space $B(\C^2)$ of operators on $\C^2$, we have that $U_0$ normalizes $I \otimes B(\C^2)$. In fact, $U_0$ defines a $*$-automorphism $\Phi$ of $B(\C^2)$ by $\sigma_0 \otimes \Phi(A)= U_0(\sigma_0 \otimes A)U_0^\dagger$ for $A \in B(\C^2)$. By the Skolem--Noether theorem \cite{skolem1927theorie, noether1933nichtkommutative}, $\Phi$ is given by conjugation by a unitary matrix $V_0$, i.e., $U_0(\sigma_0 \otimes A)U_0^\dagger = \sigma_0 \otimes V_0AV_0^\dagger$.

The operator $(\sigma_0 \otimes V_0^\dagger)U_0$ commutes with $\sigma_0 \otimes B(\C^2)$, so it must be of the form $W_0 \otimes \sigma_0$ for some unitary $W_0$, whence $U_0 = W_0 \otimes V_0$ is a product gate and not entangling. Thus, we have shown that any gate that cannot take a Pauli string of the form $\sigma_0 \otimes \sigma_i$ (for some $i \in \{1, 2, 3\}$) to a Pauli string of weight 2 is not entangling.

Now if $p_2$ assigns positive probability to any entangling gate, we must have $p_{01\to11} > 0$, so by reversibility $p_{11\to01} > 0$ and $p_{11\to11} < 1$.
\end{proof}

\subsection{Good Quantum Codes from Random Matching Circuits}\label{sec:circuit-to-code}

In this subsection, we will show how the second moment operator can be used to control the probability that our random circuit from Subsection~\ref{sec:architecture} defines a good quantum error-correcting code. We will start by describing the bound on the probability that a random circuit fails to satisfy the condition of being an error correcting code. This follows the start of the proof for Theorem 3.1 in \cite{Brown_2013}. 

Recall the following proposition from \cite{Brown_2013}, whose proof we restate in Appendix~\ref{sec:appendix-lemmas}: 

\begin{customprop}{\ref{"BF_II1_proof"}}
    A unitary $U\in\cC_n$ defines a quantum error-correcting code of distance at least $d+1$ if and only if for all $\nu_A \in \{0,1,2,3\}^k - \{0^k\}, \nu_B \in \{0,3\}^{n-k}$, and $\mu\in\{0,1,2,3\}^n$ of weight $1\leq w(\mu)\leq d$, we have $\tr[\sigma_\mu U (\sigma_{\nu_A} \otimes \sigma_{\nu_B})U^\dagger]\!=\!0$.\looseness=-1
\end{customprop}

From here, we will freely use notation from Subsection~\ref{sec:second-moment}. Suppose $U_T$ is a random Clifford obtained from $T$ layers of the random matching circuit defined in Definition~\ref{definition: random matching layer}. The second moment operator for $U_T$ is given by $M_{\mathrm{circ}, p_2}^T$, as defined in Subsection~\ref{sec:second-moment}. Following \cite{Brown_2013} we observe that \[
\PP[\tr[\sigma_\mu U_T (\sigma_{\nu_A} \otimes \sigma_{\nu_B})U_T^\dagger] \neq 0] = \frac{1}{2^n}\E\left[\left|\tr[\sigma_\mu U_T (\sigma_{\nu_A} \otimes \sigma_{\nu_B})U_T^\dagger]\right|\right]
\]
because $\tr[\sigma_\mu U_T (\sigma_{\nu_A} \otimes \sigma_{\nu_B})U_T^\dagger] \in \{0, \pm 2^n\}$. We can write \begin{align*}
\frac{1}{2^n}\E\left[\left|\tr[\sigma_\mu U_T (\sigma_{\nu_A} \otimes \sigma_{\nu_B})U_T^\dagger]\right|\right] &= \frac{1}{4^n}\E\left[\left(\tr[\sigma_\mu U_T (\sigma_{\nu_A} \otimes \sigma_{\nu_B})U_T^\dagger]\right)^2\right] \\
&= \frac{1}{4^n}\E\left[\tr[\sigma_\mu^{\otimes 2} M_{\mathrm{circ}, p_2}^T[\sigma_{\nu_A\nu_B}\otimes\sigma_{\nu_A\nu_B}]]\right] \\
&= Q^T(\nu_A\nu_B, \mu).
\end{align*}
By Proposition~\ref{"BF_II1_proof"} and the union bound, combined with the above calculation, we have that $U_T$ defines a quantum error-correcting code with probability at least \[
p_{\mathrm{code}} \geq 1 - \sum_{\substack{\nu_A \in \{0, 1, 2, 3\}^k - \{0^k\} \\ \nu_B \in \{0, 3\}^{n - k}}}\sum_{\substack{\mu \in \{0, 1, 2, 3\}^n\\ 1 \leq w(\mu) \leq d}}Q^T(\nu_A\nu_B, \mu).
\]
We split up the terms of the outer sum by $b(\nu_A\nu_B)$, the binary string with zeroes in the same location as $\nu_A\nu_B$. Then we have \[
p_{\mathrm{code}}  \geq 1 - \sum_{\substack{x_A \in \{0, 1\}^k - \{0^k\} \\ x_B \in \{0, 1\}^{n - k}}}3^{w(x_A)}\sum_{\substack{y \in \{0, 1\}^n \\ 1 \leq w(y) \leq d}}Q_0^T(x_Ax_B, y)
\]
Define a matrix \[
P(w, \ell) = \sum_{\substack{y \in \{0, 1\}^n \\ w(y) = \ell}} Q_0(x, y) \qquad\text{ for any (all) } x \in \{0, 1\}^n \text{ of weight }w.
\]
Also, set \[
P^T(w, \leq d) = \sum_{\ell = 1}^d P^T(w, \ell).
\]
To see that $P(w, \ell)$ is independent of the starting string $x$ used to define it, we observe that the distribution of a uniformly random perfect matching is invariant under composition with permutations of $[n]$. Thus, $Q_0(x, y) = Q_0(\tau x, \tau y)$ for any permutation $\tau$ of the indices of $x$ and $y$. 
Since permutations of $[n]$ act transitively on strings of a fixed weight, $\sum_{\substack{y \in \{0, 1\}^n \\ w(y) = m}} Q_0(x, y) = \sum_{\substack{y \in \{0, 1\}^n \\ w(y) = m}} Q_0(\tau x, y)$. Thus, \begin{equation}\label{eq:code-prob}
p_{\mathrm{code}} \geq 1 - \sum_{\substack{x_A \in \{0, 1\}^k - \{0^k\} \\ x_B \in \{0, 1\}^{n - k}}}3^{w(x_A)}\sum_{\substack{y \in \{0, 1\}^n \\ 1 \leq w(y) \leq d}}Q_0^T(x_Ax_B, y) \geq 1 -  \sum_{\substack{x_A \in \{0, 1\}^k - \{0^k\} \\ x_B \in \{0, 1\}^{n - k}}}3^{w(x_A)}P^T(w(x_Ax_B), \leq d)
\end{equation}
The matrix $P$ also defines a Markov chain on $[n]$, and we study it carefully in Section~\ref{sec:markov-chain}. In particular, given a matrix $K'$ of the form \eqref{eq:K'} we will produce a Markov chain with transition matrix $P$ as above, and in Section~\ref{sec:markov-chain} we will show the following theorem: 

\begin{customthm}{\ref{thm:markov-chain-main}}
Suppose $K'$ is reversible with respect to the stationary distribution $\pi' = (1/5, 1/5, 3/5)$ on $01$, $10$, $11$, and that $0 < p_{11\to11} < 1$. Let $m, \beta, \delta > 0$. There is a constant $c > 0$ depending on $\delta$, $m$, $\beta$, and $p_{11\to11}$ such that for large enough (even) $n$ and for $t \geq c\log n$ we have \[
P^t(w, \ell) \leq 4^{\delta n}\frac{3^\ell\binom{n}{\ell}}{4^n - 1} + \frac{1}{3^w\binom{n}{w}}n^{-m} + e^{-\beta n}
\]
for any integers $1 \leq w, \ell \leq n$ and \[
P^t(w, \leq d) \leq \frac{4^{\delta n}}{4^n - 1}3^d2^{nH(d/n)} + \frac{1}{3^w\binom{n}{w}}n^{-m} + e^{-\beta n}
\]
for any integers $1 \leq w\leq n$ and $1 \leq d \leq n/2$, where $H$ is the binary entropy function.
\end{customthm}

The bounds from Theorem~\ref{thm:markov-chain-main} will give us the following result, whose proof is very similar to that of \cite[Theorem III.1]{Brown_2013}:

\begin{theorem}[Universality for good quantum codes from random matching circuits]\label{thm:main-formal}
Suppose $p_2$ is a distribution on $\cC_2$ satisfying bi-invariance, positive entanglement, and reversibility (Definition~\ref{def:universality-conditions}). Let $m, \delta > 0$. There is a constant $c > 0$ depending on $\delta$, $m$, and $p_2$ such that for large enough (even) $n$, a random matching circuit (Definition~\ref{definition: random matching layer}) of depth $\lceil c\log n\rceil $ derived from $p_2$ defines an $[n, k]$ quantum error-correcting code with distance at least $d + 1$ with probability at least \[
p_{\mathrm{code}} \geq 1 - n^{-m} - 2^{k - n(1 - H(d/n) - \log_2(3)d/n - \delta)}.
\]
\end{theorem}
\begin{proof}
By Lemma~\ref{lem:K'-reversible}, since $p_2$ satisfies bi-invariance, positive entanglement, and reversibility, $K'_{p_2}$ is reversible with respect to $\pi'$ and $0 < p_{11\to11} < 1$. Thus, Theorem~\ref{thm:markov-chain-main} applies (with $\delta/3$ in place of $\delta$ and $m + 2$ in place of $m$), and for each $\beta > 0$ (to be chosen later) there is a $c > 0$ such that for $T \geq c\log n$ we have \begin{align*}
1 - p_{\mathrm{code}} &\leq \sum_{w=1}^n\sum_{\substack{x_A \in \{0, 1\}^k - \{0^k\} \\ x_B \in \{0, 1\}^{n - k} \\ w(x_Ax_B) = w}}3^{w(x_A)}P^T(w, \leq d) \\
&\leq \sum_{w=1}^n\sum_{\substack{x_A \in \{0, 1\}^k - \{0^k\} \\ x_B \in \{0, 1\}^{n - k} \\ w(x_Ax_B) = w}}3^{w(x_A)}\left(\frac{4^{\delta n/3}}{4^n - 1}3^d2^{nH(d/n)} + \frac{1}{3^w\binom{n}{w}}n^{-m - 2} + e^{-\beta n}\right).
\end{align*}
We split into two sums and bound each one separately.

To bound \[
\sum_{w=1}^n\sum_{\substack{x_A \in \{0, 1\}^k - \{0^k\} \\ x_B \in \{0, 1\}^{n - k} \\ w(x_Ax_B) = w}}3^{w(x_A)}\left(\frac{4^{\delta n/3}}{4^n - 1}3^d2^{nH(d/n)} + e^{-\beta n}\right)
\]
we use the fact that \[
\sum_{\substack{x_A \in \{0, 1\}^k - \{0^k\} \\ x_B \in \{0, 1\}^{n - k} \\ w(x_Ax_B) = w}}3^{w(x_A)} \leq \sum_{\substack{\nu_A \in \{0, 1, 2, 3\}^k - \{0^k\} \\ \nu_B \in \{0, 3\}^{n - k}}} 1 \leq 4^k2^{n-k} - 1 = 2^{n + k} - 1
\]
in the inner sum, so \begin{align*}
\sum_{w=1}^n\sum_{\substack{x_A \in \{0, 1\}^k - \{0^k\} \\ x_B \in \{0, 1\}^{n - k} \\ w(x_Ax_B) = w}}3^{w(x_A)}\left(\frac{4^{\delta n/3}}{4^n - 1}3^d2^{nH(d/n)} + e^{-\beta n}\right) &\leq n(2^{n+k} - 1)\left(\frac{4^{\delta n/3}}{4^n - 1}3^d2^{nH(d/n)} + e^{-\beta n}\right) \\
&\leq n2^{k - n(1 - H(d/n) - \log_2(3)d/n - 2\delta/3)} + ne^{-(\beta - \ln 4)n}.
\end{align*}
Now choose $\beta > \ln 4 + 1$ so that the right hand side is bounded above by $2n2^{k - n(1 - H(d/n) - \log_2(3)d/n - 2\delta/3)}$, which is bounded above by $2^{k - n(1 - H(d/n) - \log_2(3)d/n - \delta)}$ for large enough $n$. 

We can make the bound uniform in $n$ by increasing $c$.

To bound \[
\sum_{w=1}^n\sum_{\substack{x_A \in \{0, 1\}^k - \{0^k\} \\ x_B \in \{0, 1\}^{n - k} \\ w(x_Ax_B) = w}}3^{w(x_A)}\frac{1}{3^w\binom{n}{w}}n^{-m - 2}
\]
we use the fact that \begin{align*}
\sum_{\substack{x_A \in \{0, 1\}^k - \{0^k\} \\ x_B \in \{0, 1\}^{n - k} \\ w(x_Ax_B) = w}}3^{w(x_A)} &= \sum_{q = 1}^w \binom{k}{q}3^q\binom{n - k}{w - q} \\
&\leq \sum_{q = 1}^w \binom{n}{w}3^w \leq n\binom{n}{w}3^w
\end{align*}
so that \begin{align*}
\sum_{w=1}^n\sum_{\substack{x_A \in \{0, 1\}^k - \{0^k\} \\ x_B \in \{0, 1\}^{n - k} \\ w(x_Ax_B) = w}}3^{w(x_A)}\frac{1}{3^w\binom{n}{w}}n^{-m - 2} &\leq \sum_{w=1}^n n\binom{n}{w}3^w\frac{1}{3^w\binom{n}{w}}n^{-m - 2} \\ &= \sum_{w=1}^n n^{-m - 1} = n^{-m}.
\end{align*}
Combining the bounds from the two sums yields the claim in the theorem.
\end{proof} 
 
\section{Analysis of the Pauli Weight Chain}\label{sec:markov-chain}

In this section, we study a class of Markov chains $Q_0$ on nonzero binary strings of even length $n$. The input to this process is a transition kernel on strings of length $2$ given by a row-stochastic $4\times 4$ matrix of the following form: \[
K = \begin{pmatrix}
    1 & 0 & 0 & 0 \\
    0 & p_{01\to01} & p_{01\to10} & p_{01\to11} \\
    0 & p_{10\to01} & p_{10\to10} & p_{10\to11} \\
    0 & p_{11\to01} & p_{11\to10} & p_{11\to11}
\end{pmatrix}
\]
We denote by $K'$ the bottom-right $3\times3$ corner, and we require that $K'$ is reversible with respect to the stationary distribution $\pi' = (1/5, 1/5, 3/5)$ and that $0 < p_{11\to11} < 1$. The data of $K'$ and its entries will be in use throughout the rest of this section. Note that row stochasticity implies that $p_{11\to01} + p_{11\to10} = 1 - p_{11\to11}$, and stationarity of $\pi'$ implies that $p_{01\to11} + p_{10\to11} = 3 - 3p_{11\to11}$. The entry $p_{11\to11}$ will govern most of the dynamics of the processes described in this section.

One step of the process $Q_0$ is defined as follows:

\begin{algorithm}[H]
\caption{String chain $Q_0$}\label{permutation_M}
\begin{algorithmic}[1]
\REQUIRE $y\in\{0,1\}^n - \{0^n\}$
\STATE Pick a uniformly random ordered perfect matching $\mathcal{O}$ of $[n]$. 
\STATE For each $(i,j)\in\mathcal{O}$, let $y_i'y_j'$ be obtained at random from $y_iy_j$ according to the transition kernel $K$. 
\STATE Output $y'\in\{0,1\}^n$.
\end{algorithmic}
\end{algorithm} 

In words, we partition $[n] \coloneqq \{1, \dots, n\}$ into ordered pairs uniformly at random. Each pair of indices defines a two-bit string, which we evolve according to $K$. In particular, the string $00$ is almost surely not modified in this process.

\begin{remark}\label{rmk:symmetrize}
Let $\Tilde{K}'$ be the matrix $K'$ with the first and second row and column switched (i.e., exchanging $p_{01\to11}$ for $p_{10\to11}$ and so on). We note that $K'$ and $\Tilde{K}'$ yield the same string chain $Q_0$, since $(i, j)$ and $(j, i)$ appear in a uniformly random ordered perfect matching $\mathcal{O}$ with equal probability. In particular, $\frac{1}{2}(K' + \Tilde{K}')$ also yields the same string chain. When studying the string chain, we do not lose any generality by assuming \[
K' = \begin{pmatrix}
    p_{1\to1} & p_{swap} & p_{1\to2} \\
    p_{swap} & p_{1\to1} & p_{1\to2} \\
    \frac{1}{2}p_{2\to1} & \frac{1}{2}p_{2\to1} & p_{2\to2}
\end{pmatrix}
\]
where $p_{1\to1} = \frac{1}{2}(p_{01\to01} + p_{10\to10})$, $p_{swap} = \frac{1}{2}(p_{01\to10} + p_{10\to01})$, $p_{1\to2} = \frac{1}{2}(p_{01\to11} + p_{10\to11})$, and $p_{2\to1} = p_{11\to01} + p_{11\to10}$.
\end{remark}

Any random 2-qubit Clifford gate satisfying the conditions of Definition~\ref{def:universality-conditions} yields a kernel $K$ of the desired form via the second-moment operator by Lemma~\ref{lem:K'-reversible}.

We recall that the \textit{weight} $w(y)$ of a string $y \in \{0, 1\}^n$ is the number of ones in the string. The Markov chain $Q_0$ defines a \textit{weight chain} $P$ on $[n]$ as follows:

\begin{algorithm}[H]
\caption{Weight chain $P$}\label{permutation_M2}
\begin{algorithmic}[1]
\REQUIRE $w \in [n]$
\STATE Let $y \in \{0, 1\}^n$ be a string of weight $w$. 
\STATE Let $y'$ be obtained from $y$ at random according to $Q_0$.
\STATE Output $w(y')$.
\end{algorithmic}
\end{algorithm} 

We observe that $P$ is indeed a Markov chain. To see this, note that the distribution of $w(y')$ is insensitive to permutations of the string $y$ because the distribution of a uniformly random ordered perfect matching of $[n]$ is invariant under composition with permutations of $[n]$. We denote by $P^t(w, k)$ the probability of reaching $k$ from $w$ in exactly $t$ steps of $P$ starting from $w$ and set $P^t(w, [k_1, k_2]) \coloneqq \sum_{k=k_1}^{k_2} P^t(w, k)$ to be the probability of being in the range $[k_1, k_2]$ after exactly $t$ steps of $P$ starting from $w$.

\begin{lemma}[Stationary distributions]\label{lem:stationary}
Use notation from the start of this section. Assume $0 < p_{11\to11} < 1$ and $K'$ is reversible with respect to $\pi' = (1/5, 1/5, 3/5)$. Then the following is true.

The chains $Q_0$ and $P$ are irreducible, and aperiodic. The chain $Q_0$ is reversible with respect to the stationary distribution \[
\pi_0(y) = \frac{3^{w(y)}}{4^n - 1}
\]
and the chain $P$ is reversible with respect to the stationary distribution \[
\pi(w) = \frac{3^w\binom{n}{w}}{4^n - 1}.
\]
\end{lemma}
The assumptions on $K'$ will be standing assumptions throughout this section. 
\begin{proof}
To see that $Q_0$ is irreducible, we first observe that any two states with weight 1 can communicate. To see this, let $y^{(i)}$ be the string of weight 1 with 1 in the $i$'th position. Since $p_{11\to11} < 1$ we have $p_{01\to11} > 0$ or $p_{10\to11} > 0$, and $p_{11\to01} > 0$ or $p_{11\to10} > 0$. Then for $i \neq j$ we see that $y^{(j)}$ is reachable from $y^{(i)}$ in two steps by picking a partition involving the pair $(i, j)$ or $(j, i)$, replacing $y^{(i)}_iy^{(i)}_j = 10$ with $11$, and then picking another partition involving the pair $(i, j)$ or $(j, i)$ and replacing the $11$ string by $01$. 

Next we see that any string $y$ communicates with a string of weight 1. To see this, we will show that $y$ communicates with a string of weight $w(y) - 1$ if $w(y) > 1$. If $w(y)$ is even, we can pair up all the $1$s in $y$ and change exactly one of them to $10$ or $01$ (using $p_{11\to11} > 0$ and $p_{11\to01} > 0$ or $p_{11\to10} > 0$). If $w(y)$ is odd, this can be done in two steps. We necessarily have $w(y) < n$, so in the first step we can pair up all but one of the $1$s and pair the remaining one with a $0$. Then we change none of the $11$ pairs and replace the $10$ or $01$ pair with $11$ (using $p_{11\to11} > 0$ and $p_{01\to11} > 0$ or $p_{10\to11} > 0$). In the second step, the new weight is even and at least $4$, so we can pair up all the $1$s and change exactly two of them into $10$ or $01$. We observe that the same process can be performed in reverse, so any string communicates with a string of weight 1. Hence $Q_0$ is irreducible.

To see that $Q_0$ is aperiodic, we observe that any state of even weight is aperiodic; we can leave it unchanged with positive probability by pairing up all the $1$s and leaving each $11$ pair unchanged (using $p_{11\to11} > 0$).

To see that $Q_0$ is reversible with respect to $\pi_0$, let $y, z$ be two strings with $Q_0(y, z) > 0$ (i.e., the probability of reaching $z$ from $y$ in one step of $Q_0$ is positive). Let $Y$ be the random string obtained from $y$ after one step of $Q_0$ and $Z$ the random string obtained from $z$ after one step of $Q_0$. Let $\mathcal{O}_y, \mathcal{O}_z$ be uniformly random partitions of $[n]$ into ordered pairs used to evolve $y$ and $z$, respectively.

Fix a partition $O$. Let $O_1, \dots, O_{n/2}$ be the associated ordered pairs. Assume that if $O_k = (i_k, j_k)$ with $y_{i_k}y_{j_k} = 00$ then $z_{i_k}z_{j_k} = 00$ as well, and vice versa. We claim that \[
\pi_0(y)\PP[Y = z \mid \mathcal{O}_y = O] = \pi_0(z)\PP[Z = y \mid \mathcal{O}_z = O]
\]
Indeed, for $Y = z$ conditioned on $\mathcal{O}_y = O$ we need that for each $O_k = (i_k, j_k)$ we have $Y_{i_k}Y_{j_k} = z_{i_k}z_{j_k}$; similarly for $Z = y$ we need that for each $O_k$ we have $Z_{i_k}Z_{j_k} = y_{i_k}y_{j_k}$. Let $\pi' = (1/5, 1/5, 3/5)$ be the stationary distribution of $K'$ on strings $(01, 10, 11)$.

For each $k$, either we have $y_{i_k}y_{j_k} = z_{i_k}z_{j_k} = Y_{i_k}Y_{j_k} = Z_{i_k}Z_{j_k} = 00$ or, by reversibility of $K'$ with respect to $\pi'$, we have \[
\pi'(y_{i_k}y_{j_k})\PP[Y_{i_k}Y_{j_k} = z_{i_k}z_{j_k} \mid \mathcal{O}_y = O]  = \pi'(z_{i_k}z_{j_k})\PP[Z_{i_k}Z_{j_k} = y_{i_k}y_{j_k} \mid \mathcal{O}_z = O].
\]
Hence, \begin{align*}
\PP[Y = z \mid \mathcal{O}_y = O]\prod_{k=1}^{n/2}\pi'(y_{i_k}y_{j_k}) &= \prod_{k=1}^{n/2} \pi'(y_{i_k}y_{j_k})\PP[Y_{i_k}Y_{j_k} = z_{i_k}z_{j_k} \mid \mathcal{O}_y = O] \\
&= \prod_{k=1}^{n/2} \pi'(z_{i_k}z_{j_k})\PP[Z_{i_k}Z_{j_k} = y_{i_k}y_{j_k} \mid \mathcal{O}_z = O] \\
&= \PP[Z = y\mid\mathcal{O}_z = O]\prod_{k=1}^{n/2} \pi'(z_{i_k}z_{j_k}).
\end{align*}
Each term $\pi'(y_{i_k}y_{j_k})$ is proportional to $3^{w(y_{i_k}y_{j_k})}$, and similarly for $z$. Hence, we have shown \[
\PP[Y = z \mid \mathcal{O}_y = O]3^{w(y)} = \PP[Z = y\mid\mathcal{O}_z = O]3^{w(z)}.
\]
Averaging over $\mathcal{O}_y$ and $\mathcal{O}_z$ we obtain the detailed balance condition for $Q_0$: \[
\PP[Y = z]3^{w(y)} = \PP[Z = y]3^{w(z)}.
\]
Hence $Q_0$ is reversible with respect to $\pi_0$ and $\pi_0$ is stationary for $Q_0$.

Now we will transfer these nice properties from $Q_0$ to $P$.

Since $Q_0$ is irreducible and aperiodic, we immediately get that $P$ is as well. Indeed, if any two strings communicate in $Q_0$, then their weights communicate in $P$. Moreover, if any string is aperiodic in $Q_0$, then its weight is aperiodic in $P$. To show that $P$ is reversible with respect to the stationary distribution, choose two weights $w$ and $w'$ and add up the detailed balance condition for $Q_0$ over all pairs of strings of weights $w$ and $w'$, respectively to obtain $\pi(w)P(w, w') = \pi(w')P(w', w)$.\qedhere\\
\end{proof}

The goal of this section is to get bounds on $P^t(w, [1, d])$ in terms of $w$ when $t$ is large enough in terms of $n$. The strategy is similar to the proof of \cite[Theorem 4.2]{brownDecouplingRandomQuantum2015}, and we will get a very similar bound. The main idea is that when $w$ is close to the average of the stationary distribution, then $P^t(w, k)$ is not much bigger than $\pi(k)$. This part of the proof is essentially identical to the one in \cite{brownDecouplingRandomQuantum2015}. For more general $w$, we will use a hitting time argument to show that the weight chain gets close to the average of the stationary distribution within $O(\log n)$ steps with very high probability. This argument is quite different from that of \cite{brownDecouplingRandomQuantum2015}. At each step of the string chain, we will study the ordered pairs in $\mathcal{O}$ one at a time as a martingale and use concentration inequalities for martingales together with a large deviation bound to control the hitting time.

\begin{lemma}[Start of proof of {\cite[Theorem 4.2]{brownDecouplingRandomQuantum2015}}]\label{lem:bound-from-stationarity}
Suppose $|w - 3n/4| \leq \delta n$ for some $\delta \in (0, 1/16)$. Then for large enough (even) $n$ depending only on $\delta$ we have \[
P^t(w, \ell) \leq 4^{\delta n}\frac{3^\ell\binom{n}{\ell}}{4^n - 1}
\]
and \[
P^t(w, [1, d]) \leq \frac{4^{\delta n}}{4^n - 1}3^d2^{nH(d/n)}
\]
for any $t \geq 0$, $1 \leq \ell\leq n$, and $1 \leq d \leq n/2$, where $H$ is the binary entropy function.
\end{lemma}
\begin{proof}
We have $\pi(w)P^t(w, \ell) \leq \sum_{i=1}^n \pi(i)P^t(i, \ell) = \pi(\ell)$, where the second equality follows from stationarity of $\pi$. Thus, we have \[
P^t(w, \ell) \leq \frac{\pi(\ell)}{\pi(w)} = \frac{4^n - 1}{3^w\binom{n}{w}}\frac{3^\ell\binom{n}{\ell}}{4^n - 1}.
\]
Now by Stirling's formula we have \[
\binom{n}{\lfloor 3n/4\rfloor }3^{3n/4} \geq \frac{C}{\sqrt{n}}4^n
\]
for some constant $C > 0$. We wish to bound $3^w\binom{n}{w}$ when $w$ is close to $3n/4$. Indeed when $w \leq 3n/4$ we have \begin{align*}
\binom{n}{w}3^w &= \binom{n}{\lfloor 3n/4 \rfloor }3^{3n/4}3^{w-3n/4}\prod_{j=0}^{\lfloor 3n/4 \rfloor -w}\frac{\lfloor3n/4\rfloor - j}{\lceil n/4\rceil + j + 1} \\
&= \binom{n}{\lfloor 3n/4\rfloor }3^{3n/4}\prod_{j=0}^{\lfloor 3n/4\rfloor -w}\frac{\lfloor3n/4\rfloor - j}{3\lceil n/4\rceil + 3j + 3}.
\end{align*}
When $w \geq (3/4 - \delta)n$, each of the at most $\delta n$ terms in the product on the right hand side is bounded below by $\frac{3n/4 - \delta n - 1}{3n/4 + 3\delta n + 6} = 1 - \frac{16\delta n + 28}{3n + 12\delta n + 24}$. When $n$ is large enough we can bound this below by $1 - 16\delta/3$ so that \[
\binom{n}{w}3^w \geq \frac{C}{\sqrt{n}}4^n(1 - 16\delta/3)^{\delta n}.
\]
Similarly, when $3n/4 \leq w \leq (3/4 + \delta)n$ we have \begin{align*}
\binom{n}{w}3^w &= \binom{n}{\lfloor 3n/4\rfloor }3^{3n/4}3^{w - 3n/4} \prod_{j=0}^{w - \lfloor 3n/4\rfloor } \frac{\lceil n/4 \rceil - j}{\lfloor 3n/4\rfloor + j + 1} \\
&= \binom{n}{3n/4}3^{3n/4} \prod_{j=0}^{w - \lfloor 3n/4\rfloor} \frac{3\lceil n/4\rceil - 3j}{\lfloor 3n/4\rfloor + j + 1}. 
\end{align*}
Each of the at most $\delta n$ terms in the product on the right hand side is bounded below by $\frac{3n/4 - 3\delta n}{3n/4 + \delta n + 1} = 1 - \frac{16\delta  + 4/n}{3 + 4\delta + 4/n}$, which we can again bound from below by $1 - 16\delta/3$. Thus we get the same bound for $\binom{n}{w}3^w$ as before.

If $\delta < 1/16$ then $1 - 16\delta/3 > 2/3$, so $\frac{C}{\sqrt{n}}4^n(1 - 16\delta/3)^{\delta n} \geq 4^{(1 - \delta)n}$ for sufficiently large $n$. Thus \[
P^t(w, \ell) \leq \frac{4^n - 1}{4^{(1 - \delta)n}}\frac{3^\ell\binom{n}{\ell}}{4^n - 1} \leq 4^{\delta n}\frac{3^\ell\binom{n}{\ell}}{4^n - 1}.
\]
This concludes the proof of the first part of the claim.

The second claim follows from the bound $\sum_{\ell=0}^d \binom{n}{\ell} \leq 2^{nH(d/n)}$.\qedhere\\
\end{proof}

To get from the conclusion of Lemma~\ref{lem:bound-from-stationarity} to a bound for any starting $w$, we show that the Markov chain driven by $P$ reaches the interval $[(3/4 - \delta)n, (3/4 + \delta)n]$ quickly. For $t \geq 0$ and $w \in [n]$ let $X_t(w)$ be the state of the weight chain with $X_0(w) \coloneqq w$. We will write $X_t$ if the initial state is irrelevant. For an interval $[\ell_1, \ell_2]$ we define a random variable (\textit{hitting time}) \[
T_{hit}(w, [\ell_1, \ell_2]) \coloneqq \min \{t \geq 0 \mid X_t(w) \in [\ell_1, \ell_2]\}.
\]
We have \begin{align*}
    P^t(w, \ell) &= \PP[X_t(w) = \ell] \\
    &\leq \PP[X_t(w) = \ell, X_s \in [(3/4 - \delta)n, (3/4 + \delta)n]\text{ for some }0\leq s \leq t] \\
    &\qquad +\PP[X_s \notin [(3/4 - \delta)n, (3/4 + \delta)n]\text{ for all }0\leq s \leq t] \\
    &\leq 4^{\delta n}\frac{3^\ell\binom{n}{\ell}}{4^n - 1} + \PP[T_{hit}(w, [(3/4 - \delta)n, (3/4 + \delta)n]) > t].
\end{align*}
Similarly we have \[
P^t(w, [1, d]) \leq \frac{4^{\delta n}}{4^n - 1}3^d2^{nH(d/n)} + \PP[T_{hit}(w, [(3/4 - \delta)n, (3/4 + \delta)n]) > t].
\]
The next subsections will be concerned with bounding the second term in the sum on the right hand side when $t$ is large enough. The consequence of the work in those subsections (see Proposition~\ref{prop:hitting-time}) is the following theorem:

\begin{theorem}\label{thm:markov-chain-main}
Suppose $K'$ is reversible with respect to the stationary distribution $\pi' = (1/5, 1/5, 3/5)$ on $01$, $10$, $11$, and that $0 < p_{11\to11} < 1$. Let $m, \beta, \delta > 0$. There is a constant $c > 0$ depending on $\delta$, $m$, $\beta$, and $p_{11\to11}$ such that for large enough (even) $n$ and for $t \geq c\log n$ we have \[
P^t(w, \ell) \leq 4^{\delta n}\frac{3^\ell\binom{n}{\ell}}{4^n - 1} + \frac{1}{3^w\binom{n}{w}}n^{-m} + e^{-\beta n}
\]
for any integers $1 \leq w, \ell \leq n$ and \[
P^t(w, \leq d) \leq \frac{4^{\delta n}}{4^n - 1}3^d2^{nH(d/n)} + \frac{1}{3^w\binom{n}{w}}n^{-m} + e^{-\beta n}
\]
for any integers $1 \leq w\leq n$ and $1 \leq d \leq n/2$, where $H$ is the binary entropy function.
\end{theorem}
Note that we do not assume $\delta < 1/16$, since otherwise we can shrink $\delta$ as desired.

\subsection{Martingales and One-Step Probability}

Let $\delta \in (0, 1/16)$ as in Lemma~\ref{lem:bound-from-stationarity}. In this section, we begin the task of bounding the hitting time of the weight chain $X_t$ by showing that when $X_t < (3/4 - \delta)n$, it grows rather quickly (in fact, exponentially fast) with high probability. The goal of this section is to study the distribution of $X_{t+1}$ conditional on $X_t$. To do this, we introduce a well-known concentration inequality.

Recall that a sequence of real random variables $Y_0, Y_1, \dots$ is a \textit{martingale} if $\E[|Y_i|] < \infty$ and $\E[Y_{i+1} \mid Y_1, \dots, Y_i] = Y_i$. We have the following concentration inequality for martingales:

\begin{lemma}[Azuma--Hoeffding inequality]\label{lem:hoeffding}
Suppose $Y_0, Y_1, \dots$ is a martingale and $|Y_i - Y_{i-1}| \leq c_i$ almost surely for all $i$. Then for all positive integers $N$ and $\lambda > 0$ we have \[
\PP[Y_N - Y_0 \geq \lambda] \leq \exp\left(\frac{-\lambda^2}{2\sum_{i=1}^N c_i^2}\right) 
\]
and \[
\PP[Y_N - Y_0 \leq -\lambda] \leq \exp\left(\frac{-\lambda^2}{2\sum_{i=1}^N c_i^2}\right).
\]
\end{lemma}

We will use this inequality to prove two tail bounds (Lemma~\ref{lem:large-jumps} and Lemma~\ref{lem:contraction}) which we will combine to get our hitting time results. 

\begin{lemma}\label{lem:hoeffding-for-weight}
Let $X_0, X_1, \dots$ be the weight chain. For any $\lambda > 0$ we have \[
\PP[X_{t+1} - \E[X_{t+1} \mid X_t] \geq \lambda \mid X_t] \leq \exp\left(-\frac{\lambda^2}{32X_t}\right)
\]   
and \[
\PP[X_{t+1} - \E[X_{t+1} \mid X_t] \leq -\lambda \mid X_t] \leq \exp\left(-\frac{\lambda^2}{32X_t}\right).
\]  
\end{lemma}
\begin{proof}
Condition on $X_t = w$.

Let $W \in \{0, 1\}^n$ be a fixed string of weight $w$ and let $\mathcal{O}$ be a uniformly random partition of $[n]$ into ordered pairs. Let $W'$ be the random string obtained from $W$ by evolving according to the string process, so $X_{t+1} = w(W')$.

We define a sequence of random variables $Y_0, \dots, Y_w$ as follows. For $i = 1, \dots, w$ let $\mathrm{O}_i$ be the ordered pair in $\mathcal{O}$ containing the $i$th bit 1 in $W$. Note the $\mathcal{O}_i$ may not be distinct. When $\mathrm{O}_i$ is first revealed, we also reveal the local transition applied to the pair. Let $\mathcal{F}_i$ be the sigma algebra generated by $W$ and all information revealed up to step $i$, and define \[
Y_i \coloneqq \E[X_{t+1} \mid \mathcal{F}_i] \qquad \text{ for }0 \leq i \leq w.
\]
In particular, $Y_0 = \E[X_{t+1} \mid W] = \E[X_{t+1} \mid X_t=w]$ and $Y_w = X_{t+1}$ because each pair in $\mathcal{O}$ that is not exposed during this process corresponds to input 00 and contributes deterministically zero. By the tower property,
\[
    \E[Y_i \mid \mathcal{F}_{i-1}] = \E[\E[X_{t+1} \mid \mathcal{F}_i] \mid \mathcal{F}_{i-1}] = \E[X_{t+1} \mid \mathcal{F}_{i-1}] = Y_{i-1}
\]
Thus the sequence $\{Y_i\}$ is a martingale with respect to the filtration $(\mathcal{F}_i)_{i=0}^w$.

Let $c_i = |Y_i - Y_{i-1}|$. If $\mathrm{O}_i = \mathrm{O}_j$ for some $j < i$, then $\mathcal{F}_{i} = \mathcal{F}_{i-1}$, so $c_i = 0$.  Otherwise, revealing one new pair and its local transition outcome changes the conditional expectation of $X_{t+1}$ by at most $4$, because changing the revealed pair can be coupled by modifying at most two output pairs. Indeed, condition on the previously exposed pairs and let $v$ be the next unexposed active vertex. If one completion pairs $v$ with $u$ and another pairs $v$ with $u'$, then the two random completions can be coupled by the standard matching switch: if $u'$  is paired to $z$ in the first completion, replace the pairs $(v,u)$ and $(u',z)$ by $(v,u')$ and $(u,z)$. Thus the coupled completions differ on at most two pairs, and since each pair
contributes weight at most $2$, the conditional expectation changes by at
most $4$. Hence, \[\sum_{i=1}^w c_i^2 \leq 16w.\] By the Azuma--Hoeffding inequality (Lemma~\ref{lem:hoeffding}) we obtain the desired result.\qedhere\\
\end{proof}

\begin{lemma}[Expected drift]\label{lem:expected-drift}
Let $X_0, X_1, \dots$ be the weight chain. Let \[
E(\alpha) = \alpha + 2\frac{n}{n - 1}(1 - p_{11\to11})\alpha\left(\frac{3}{4} - \alpha\right).
\] 
We have \[
nE(w/n) \leq \E[X_{t+1}\mid X_t = w] \leq nE(w/n) + 1.
\]
\end{lemma}
\begin{proof}
Let $W$ be a string of weight $w$ and $\mathcal{O}$ a random partition of $W$ into ordered pairs. By linearity of expectation, we may compute $\E[X_{t+1} - X_t \mid X_t]$ as the sum of expected changes in weight due to each pair in $\mathcal{O}$.

A random ordered pair of bits in $W$ is $01$ or $10$ with probability $2\frac{n - w}{n}\frac{w}{n - 1}$ and $11$ with probability $\frac{w}{n}\frac{w - 1}{n - 1}$. Thus the expected change in weight for each ordered pair is \[
\frac{n - w}{n}\frac{w}{n - 1}(p_{01\to11} + p_{10\to11}) - \frac{w}{n}\frac{w - 1}{n - 1}(p_{11\to01} + p_{11\to10}).
\]
By stochasticity we have $p_{11\to01} + p_{11\to10} = 1 - p_{11\to11}$, and since $(1/5, 1/5, 3/5)$ is stationary for $K'$, we have $p_{01\to11} + p_{10\to11} = 3 - 3p_{11\to11}$. The expected change in weight for one step is \begin{align*}
\E[X_{t+1} - X_t \mid X_t = w] &= \frac{n}{2}\left(\frac{n - w}{n}\frac{w}{n - 1}(p_{01\to11} + p_{10\to11}) - \frac{w}{n}\frac{w - 1}{n - 1}(p_{11\to01} + p_{11\to10})\right)\\ &= \frac{w}{2(n - 1)}\left((n - w)(3 - 3p_{11\to11}) - (w - 1)(1 - p_{11\to11})\right) \\
&= \frac{w}{2(n - 1)}(1 - p_{11\to11})(3n - 4w + 1)
\end{align*}
The result follows from $\frac{w}{2(n-1)} \leq 1$.\qedhere\\
\end{proof}

\subsection{Hitting Time to Large Enough Weight}\label{sect:wide-hitting-time}

In this subsection, we concern ourselves with the ``low-weight'' regime $1 \leq w < n/2$. We show that, with high probability, it only takes logarithmically many steps for the weight chain to exit this regime. This is the first half of our hitting time result. In Subsection~\ref{sect:narrow-hitting-time} we will show that outside this regime the weight chain very quickly converges to any small linear-width interval around $3n/4$.

\begin{lemma}[Large jumps from low weight]\label{lem:large-jumps}
Let $\gamma$ be as in Lemma~\ref{lem:expected-drift} and let $X_0, X_1, \dots$ be the weight chain. Suppose $1 \leq w < n/2$. Then \[
\PP[X_{t+1} \leq (1 + \gamma)w \mid X_t = w] \leq e^{-\gamma^2w/32}
\]
where $\gamma = (1 - p_{11\to11})/4$. 
\end{lemma}
\begin{proof}
By Lemma~\ref{lem:expected-drift} we have \[
\E[X_{t+1} \mid X_t = w] \geq w + 2\frac{n}{n - 1}(1 - p_{11\to11})w\left(\frac{3}{4} - \frac{w}{n}\right) > \left(1 + \frac{1 - p_{11\to11}}{2}\right)w = (1 + 2\gamma)w.
\]
Applying the second half of Lemma~\ref{lem:hoeffding-for-weight} with $\lambda = \gamma w$ gives the desired result.
\end{proof}

Lemma~\ref{lem:large-jumps} shows that the weight chain grows exponentially fast with high probability. Proposition~\ref{prop:hitting-from-below} uses this fact to show a logarithmic hitting time to $n/2$ for the weight chain. Recall that we defined \[
T_{hit}(w, [n/2, n]) \coloneqq \min \{t \geq 0 \mid X_t(w) \in [n/2, n]\}.
\]
The strategy for proving Proposition~\ref{prop:hitting-from-below} is to instead control the time it takes for the chain to escape an interval $[a, n/2)$ from either side. It will turn out that escaping from the left side is much less likely than escaping from the right side. Define the \textit{escape time} $T_{esc}(w, [a, n/2))$ by \[
T_{esc}(w, [a, n/2)) \coloneqq \min\{t \geq 0 \mid X_t(w) \notin [a, n/2)\}.
\] 
Before we prove Proposition~\ref{prop:hitting-from-below} we need one more lemma that controls this escape time:

\begin{lemma}\label{lem:escape-time}
Let $X_0, X_1, \dots$ be the weight chain. Let $m > 0$. There is a constant $c$ depending on $m$ and $p_{11\to11}$ such that for $a \in [1, n/2)$ and $w \in [a, n/2)$ we have \[
\PP[T_{esc}(w, [a, n/2)) > c\log n] \leq n^{-ma}.
\]
\end{lemma}
\begin{proof}
Let $\gamma$ be as in Lemma~\ref{lem:large-jumps}. 

The strategy for this proof is to consider the quantity \[
R_t \coloneqq \begin{cases}
    (n/2X_t)^{\theta a} &\text{if } X_t \geq a \\
    0 &\text{otherwise.}
\end{cases}
\]
for a small positive parameter $\theta$. This quantity is large when $X_t$ is small and larger than $a$; it is at most 1 exactly when $X_t$ is outside the range $[a, n/2)$. We want to show that it drops below $1$ quickly. We will do this by controlling $\E[R_{t+1} \mid R_t]$; we will show that it decays exponentially, so that eventually $\E[R_T \mid R_0] < 1$ and we can conclude using Markov's inequality. 

By Lemma~\ref{lem:large-jumps} we have \[
\PP[X_{t+1} \leq (1 + \gamma)X_t \mid X_t] \leq e^{-\gamma^2X_t/32}.
\]
Condition on $X_t = u \in [a, n/2)$. Since $R_{t+1} \leq (n/2a)^{\theta a}$ we have \begin{equation}\label{eq:rt-exp}
\left(\frac{2u}{n}\right)^{\theta a}\E\left[R_{t+1}\ \middle|\ X_t = u\right] \leq (1 + \gamma)^{-\theta a} + e^{-\gamma^2u/32}\left(\frac{u}{a}\right)^{\theta a}.
\end{equation}
At this point we will split into two regimes by picking a constant $a_0$; we will handle the cases $a \geq a_0$ and $a < a_0$ separately. We start by working in the large $a$ regime; we will pick $a_0$ later to make the proof work out. 

Since $a \leq u$, we have $a\log(u/a) \leq u/e$. By choosing $\theta$ sufficiently small, say $\theta \leq \gamma^2e/64$, we have \[
e^{-\gamma^2u/32}\left(\frac{u}{a}\right)^{\theta a} \leq e^{-\gamma^2u/64} \leq e^{-\gamma^2a/64}
\]
so that \[
\left(\frac{2u}{n}\right)^{\theta a}\E\left[R_{t+1}\ \middle|\ X_t = u\right] \leq e^{-\theta\log(1+\gamma)a} + e^{-\gamma^2a/64}
\]
Let $2C = \min\{\theta\log(1 + \gamma), \gamma^2/64\} > 0$. Choose $a_0 = (\log 2)/C$. Then for $a \geq a_0$, we have $e^{-\theta\log(1 + \gamma)a}, e^{-\gamma^2a/64} \leq \frac{1}{2}e^{-Ca}$. It follows that for $a \geq a_0$, we have \[
\E[R_{t+1} \mid R_t] \leq e^{-Ca}R_t
\]
whence \[
\E[R_T \mathbbm{1}\{T_{esc}(w,[a,n/2))>T\}  \mid X_0 = w] \leq e^{-CaT}\left(\frac{n}{2w}\right)^{\theta a} \leq e^{-CaT}n^{\theta a}.
\]
If $T \geq c\log n$ then the right hand side is bounded above by $n^{-(Cc - \theta)a}$. Choose $c$ large enough that when $a \geq a_0$ we have $Cc-\theta>0$. Then for large enough $n$ depending on $(Cc - \theta)a_0$ we have, by Markov's inequality, \[
\PP[T_{esc}(w, [a, n/2)) > c\log n] \leq\PP[R_T\mathbbm{1}\{T_{esc}(w,[a,n/2))>T\}\geq 1 \mid X_0=w] \leq n^{-(Cc - \theta)a} \leq n^{-\frac{Cc - 2\theta}{2}a}.
\]
Now we consider the regime $a < a_0$. We will split into two cases. Our starting point is \eqref{eq:rt-exp}. Let $u_0$ be large enough that $e^{-\gamma^2u_0/64} \leq \frac{1}{2}(1 - (1 + \gamma)^{-\theta a})$. Note $u_0$ is a constant depending only on $\gamma$ and $\theta$ (which itself depends only on $\gamma$), and for $u \geq u_0$ we have \[
\left(\frac{2u}{n}\right)^{\theta a}\E\left[R_{t+1}\ \middle|\ X_t = u\right] \leq e^{-\theta\log(1+\gamma)a} + e^{-\gamma^2u/64} \leq \frac{1}{2}(1 + (1 + \gamma)^{-\theta a}) < 1.
\]
Then there is a $C' > 0$ bounded away from 0 and depending only on $\gamma$, $\theta$, and $a_0$ such that if $u \geq u_0$ and $a < a_0$ we have \[
\E[R_{t+1} \mid X_t = u] \leq e^{-C'a}\left(\frac{n}{2u}\right)^{\theta a}.
\]
Now we consider the case $a \leq u \leq U \coloneqq a_0 + u_0$. This is now a bounded problem, where we work with strings of weight that is very low relative to the length of the string. In such cases, when taking steps of the string process, we expect no $11$ pairs to appear among the two-bit strings coming from the matching $\mathcal{O}$. Indeed, we have \[
\PP[\mathcal{O} \text{ has a }11\text{ pair}] \leq \frac{\binom{U}{2}}{n - 1} \leq \frac{U^2}{n}.
\]

Conditioned on $X_{t} = u$ and on no $11$ pairs appearing in $\mathcal{O}$, we have $X_{t+1} = u + Y_u$, where $Y_u$ follows a binomial distribution with $u$ trials and with success probability $p_{1\to2}$ (see Remark~\ref{rmk:symmetrize} for notation). In particular, with probability at least $p_{1\to2}$ we have $Y_u \geq 1$ so that $X_{t+1} \geq u + 1$. Thus \begin{align*}
\E\left[\left(\frac{u}{X_{t+1}}\right)^{\theta a} \ \middle|\ \mathcal{O} \text{ has no }11\text{ pair}, X_t = u\right] &\leq 1 - p_{1\to2} + p_{1\to2}\left(\frac{u}{u + 1}\right)^{\theta a} = 1 - p_{1\to2}\left(1 - \left(\frac{u}{u + 1}\right)^{\theta a}\right) \\
&= 1 - p_{1\to2}\left(1 - \exp(-\theta a \log(1 + 1/u))\right) \\
&\leq 1 - p_{1\to2}\left(1 - \exp\left(-\frac{\theta a}{2U}\right)\right)
\end{align*}
Since $\log(1+1/u)\ge 1/(2U)$ and $a\le U$, we have $1-\exp(-\theta a/(2U))\ge \theta a/(4U)$, so \[
\E\left[\left(\frac{u}{X_{t+1}}\right)^{\theta a} \ \middle|\ \mathcal{O} \text{ has no }11\text{ pair}, X_t = u\right] \leq 1 - p_{1\to2}\frac{\theta a}{4U} \leq \exp\left(-p_{1\to2}\frac{\theta a}{4U}\right).
\]
We always have $\left(\frac{u}{X_{t+1}}\right)^{\theta a} \leq \left(u\right)^{\theta a} \leq U^{\theta a}$ so \[
\E\left[\left(\frac{u}{X_{t+1}}\right)^{\theta a} \ \middle|\ X_t = u\right] \leq \exp\left(-p_{1\to2}\frac{\theta a}{4U}\right) + \frac{U^2}{n}U^{\theta a}.
\]
Then for $n$ large enough, there is a constant $C'' > 0$ bounded away from 0 and depending on $p_{1\to2}$ (hence on $\gamma$), $\theta$, and $U$ such that \[
\E[R_{t+1} \mid R_t] \leq e^{-C''a}R_t
\]
whenever $a < a_0$ (regardless of $u$). Now the same argument as before shows $\PP[T_{esc}(w, [a, n/2)) > c\log n] \leq n^{-\frac{(C''c - 2\theta)a}{2}}$ when $a < a_0$, and we are done.\qedhere\\
\end{proof}

Now we use Lemma~\ref{lem:escape-time} to control how long it takes for the chain to reach $[n/2, n]$. We prove a result that is somewhat stronger than control on $T_{hit}(w, [n/2, n])$; instead, we bound the probability that any trajectory starting at $w$ spends too long in $[1, n/2)$ at any particular point in the trajectory. We will need this to argue that if the chain manages to reach $[n/2, n]$ and then drops out of $[n/2, n]$ later, it can still return quickly.

\begin{proposition}[Logarithmic hitting time to high weight]\label{prop:hitting-from-below}
Let $X_0, X_1, \dots$ be the weight chain. For every $t\geq 0$, define \[
T_{hit}^{\geq t}(w, [n/2, n]) = \min\{s \geq t \mid X_s(w) \in [n/2, n]\} - t.
\]
Let $m > 0$. Then, there exists constants $c, C > 0$ depending on $m$ and $p_{11\to11}$ such that for large enough (even) $n$ depending on $m$, we have \[
\PP[T_{hit}^{\geq t}(w, [n/2, n]) > c\log n] \leq \frac{1}{3^w\binom{n}{w}}n^{-m}
\]
for any $1 \leq w \leq n$.
\end{proposition}
\begin{proof}
Let $\gamma$ be as in Lemma~\ref{lem:large-jumps}. 

The main idea of the proof is that with high probability depending on $a$, the weight chain starting at $w < n/2$ does not stay in any range $[a, n/2)$ for more than about $\log n$ steps. We will use the reversibility of the weight chain to control that the probability that a $T$-step trajectory of the weight chain has minimum at most $a$ when $T \approx \log n$, then combine these results to show that the weight chain starting at $w < n/2$ does not stay in the range $[1, n/2)$ for more than about $\log n$ steps.

Let \[
A(w, [t, t + T]) = \min\{X_s(w) \mid t \leq s \leq t + T\}\] and \[\sigma_a^{\geq t}(w) = \min\{s \geq t \mid X_s(w) = a,\}
\]
and define the \textit{escape time} $T_e(w, [a, n/2))$ by \[
T_e^{\geq t}(w, [a, n/2)) \coloneqq \min\{s \geq t \mid X_s(w) \notin [a, n/2)\}.
\]
Then, we have \[
\PP[T_{hit}^{\geq t}(w, [n/2, n]) > T] = \sum_{a = 1}^{n/2 - 1} \PP[T_{hit}^{\geq t}(w, [n/2, n]) > T\text{ and } A(w, [t, t + T]) = a].
\]
We will bound each summand separately. Consider the scenario where $A(w, [t, t + T]) = a$ and the weight chain hits $a$ within $T/2$ steps after time $t$. In that case, in order for $T_{hit}^{\geq t}(w, [n/2, n]) > t + T$ we would need the chain to stay within $[a, n/2)$ for at least $T/2$ steps. 

Thus for $t \leq s \leq t + T/2$ we have \[
\PP[T_{hit}^{\geq t}(w, [n/2, n]) > T, A(w, [t, t + T]) = a, \sigma_a = s] \leq P^s(w, a)\PP[T_{esc}(a, [a, n/2)]) > T/2].
\] 
By reversibility of the weight chain (Lemma~\ref{lem:stationary}) we have \[
P^s(w, a) = \frac{\pi(a)}{\pi(w)}P^s(a, w) \leq \frac{\pi(a)}{\pi(w)},
\]
whence \[
\PP[T_{hit}^{\geq t}(w, [n/2, n]) > T, A(w, [t, t + T]) = a, \sigma_a = s] \leq \frac{\pi(a)}{\pi(w)}\PP[T_{esc}(a, [a, n/2)]) > T/2].
\]
On the other hand, suppose $\sigma_a = s > t + T/2$. Then for the first $s - t$ steps of the weight chain after time $t$, it is confined to $[a, n/2)$. We can consider the reversed chain, which hits $a$ at time $t + T - s$ and then stays in $[a, n/2)$ for $s$ steps before reaching $w$. By reversibility, the probability of this occurring is bounded above by \[
\PP[T_{hit}^{\geq t}(w, [n/2, n]) > T, A(w, [t, t + T]) = a, \sigma_a = s] \leq \frac{\pi(a)}{\pi(w)}\PP[T_{esc}(a, [a, n/2)]) > T/2].
\] 
Averaging over the $\sigma_a = s$ conditions yields \[
\PP[T_{hit}^{\geq t}(w, [n/2, n]) > T, A(w, [t, t + T]) = a] \leq \frac{\pi(a)}{\pi(w)}(T + 1)\PP[T_{esc}(a, [a, n/2)]) > T/2].
\]
For $m > 0$ we can choose $c$ such that by Lemma~\ref{lem:escape-time} we have $\PP[T_{esc}(a, [a, n/2)) > c\log n/2] \leq n^{-ma}$. Thus, \[
\PP[T_{hit}^{\geq t}(w, [n/2, n]) > c\log n\text{ and } A(w, [t, t + T]) = a] \leq \frac{\pi(a)}{\pi(w)}(c\log n + 1)n^{-ma}
\]
Summing over possible minima $a$ yields 
\begin{align*}
\PP[T_{hit}^{\geq t}(w, [n/2, n]) > c\log n] &\leq \sum_{a = 1}^{n/2 - 1} \frac{\pi(a)}{\pi(w)}n^{-ma} \\
&= \frac{c\log n + 1}{3^w\binom{n}{w}}\sum_{a = 1}^{n/2 - 1} 3^a\binom{n}{a}n^{-ma} \end{align*}
Since $3^a\binom{n}{a} \leq \left(\frac{3en}{a}\right)^a \leq (3e)^an^a$ we have \begin{align*}
\PP[T_{hit}^{\geq t}(w, [n/2, n]) > c\log n] &\leq \frac{c\log n + 1}{3^w\binom{n}{w}}\sum_{a = 1}^{n/2 - 1} (3en^{1-m})^a \\ &\leq \frac{c\log n + 1}{3^w\binom{n}{w}} \frac{(3en)^{1-m}}{1 - (3en)^{1-m}}.
\end{align*}
By replacing $m$ by a larger number (say, $2m + 1$) and then taking $n$ large enough, we obtain the desired result.\qedhere\\
\end{proof}

\subsection{Hitting Time to a Narrow Weight Range}\label{sect:narrow-hitting-time}

In this subsection, we show that, starting from linear weight, the weight chain converges to a narrow interval around $3n/4$ very quickly, in a constant number of steps, with high probability. The strategy is to study the dynamics of the function $E(\alpha)$ defined in Lemma~\ref{lem:expected-drift}; we observe that it has an attracting fixed point at $3/4$, and we use the concentration bounds from Lemma~\ref{lem:hoeffding-for-weight} to show that the weight chain's dynamics are very close to the deterministic dynamics of $E(\alpha)$. The first step is to show that a single step of the weight chain is likely to contract the distance between $X_t$ and $X_{t+1}$ by a constant factor.

\begin{lemma}[Contraction toward $3n/4$ from high weight]\label{lem:contraction}
Let $X_0, X_1, \dots$ be the weight chain. Suppose $0 < p_{11\to11} < 1$ and $w \in [n/2, n]$ with $|w - 3n/4| > \delta n$. Let $n$ be a sufficiently large even-valued integer depending on $p_{11\to11}$. Then \[
\PP\left[\left|X_{t+1} - \frac{3n}{4}\right| > \frac{1 + \sigma}{2}\left|w - \frac{3n}{4}\right|\ \middle|\ X_t = w\right] \leq 2\exp\left(-\frac{(1 - \sigma)^2\delta^2 n}{128}\right)
\]
where $\sigma \coloneqq \max\{p_{11\to11}, 1 - p_{11\to11}\} < 1$.  
\end{lemma}
\begin{proof}
For any $\alpha \in [1/2, 1]$ we have \[
\left|\frac{E(\alpha) - 3/4}{\alpha - 3/4}\right| = \left|1 - 2\frac{n}{n - 1}(1 - p_{11\to11})\alpha\right|
\]
Choose $n$ large enough such that $\frac{n}{n - 1}(1 - p_{11 \to 11}) \leq 1 - p_{11\to11}/2$ so that \[
(1 - p_{11\to11}) \leq 2\frac{n}{n - 1}(1 - p_{11\to11})\alpha \leq 2 - p_{11\to11}.
\]
Then \[
\left|\frac{E(\alpha) - 3/4}{\alpha - 3/4}\right| \leq \left|1 - 2\frac{n}{n - 1}(1 - p_{11\to11})\alpha\right| \leq \max\{p_{11\to11}, 1 - p_{11\to11}\} 
\]
and so $|E(\alpha) - 3/4| \leq \sigma|\alpha - 3/4|$. 

Let $\varepsilon = \frac{\delta(1 - \sigma)}{2}$. Then $\sigma|\alpha - 3/4| + \varepsilon \leq \frac{1 + \sigma}{2}|\alpha - 3/4|$, and if $|X_{t+1}/n - E(w/n)| \leq \varepsilon$ then \[\left|\frac{X_{t+1}}{n} - \frac34\right| \leq \frac{1 + \sigma}{2} \cdot \left|\frac wn - \frac34\right|.\] Hence \begin{align*}
\PP\left[\left|\frac{X_{t+1}}{n} - \frac{3}{4}\right| > \frac{1 + \sigma}{2}\left|\frac{w}{n}- \frac{3}{4}\right|\ \middle|\ X_t = w\right] &\leq \PP[|X_{t+1}/n - E(w/n)| > \varepsilon \mid X_t = w] \\
&\leq 2\exp\left(-\frac{\varepsilon^2n}{32}\right)
\end{align*}
as we wanted, where the second inequality comes from both halves of Lemma~\ref{lem:hoeffding-for-weight} with $\lambda = \varepsilon n$. \qedhere\\   
\end{proof}

Using Lemma~\ref{lem:contraction}, it is straightforward to give an estimate of the hitting time we want using the union bound:

\begin{lemma}\label{lem:hitting-time-narrow}
There is a constant $N$ depending on $\delta$ and $p_{11\to11}$ such that for large enough (even) $n$ we have that \[
\PP[T_{hit}(w, [(3/4 - \delta)n, (3/4 + \delta)n]) > N] \leq 2N\exp\left(-\frac{(1 - \sigma)^2\delta^2 n}{128}\right)
\]    
for all $w \in [n/2, n]$.
\end{lemma}
\begin{proof}
In order to get from a distance of at most $n/4$ to a distance of at most $\delta n$ one needs to take $N = \lceil\frac{\log(\delta/4)}{\log((1 + \sigma)/2)} \rceil$ $\frac{1 + \sigma}{2}$-contraction steps, where $\sigma$ is as in Lemma~\ref{lem:contraction}. Then, by the union bound, the right hand side of the lemma is an upper bound on the probability that at least one of the first $N$ steps of the chain starting at $w$ is not a contraction step, which proves the lemma.\qedhere\\
\end{proof}

Finally, we want to combine Lemma~\ref{lem:hitting-time-narrow} with Proposition~\ref{prop:hitting-from-below} to get control over the hitting time to $[(3/4 - \delta)n, (3/4 + \delta)n]$ from anywhere. However, the exponential error term in Lemma~\ref{lem:hitting-time-narrow} is not good enough for our purposes. Instead, we need to be able to control the exponential rate of decay of this error term. To do this, we allow the chain enough time to spend a substantial amount of time in $[n/2, n]$ and get more chances at hitting the target interval. To ensure that the chain spends enough time in $[n/2, n]$, we will need an upper bound for the probability of leaving this interval:

\begin{lemma}\label{lem:rare-downcross}
Let $X_0, X_1, \dots$ be the weight chain.  Assume $0 < p_{11 \to 11}<1$. For sufficiently large (even) $n$, uniformly over all $w \in [n/2,n]$, we have
\[
    \Pr[X_{t+1}<n/2\mid X_t=w] \le \exp\left( -\frac{(1 - \sigma)^2n}{512} \right)
\]
where $\sigma = \max\{p_{11\to11}, 1 - p_{11\to11}\}$.
\end{lemma}
\begin{proof}
Fix $w \in [n/2,n]$.  By Lemma~\ref{lem:expected-drift}, \[\E[X_{t+1}\mid X_t=w]\ge n E(w/n).\]
Also, by the definition of $E(\alpha)$,
\[
    E\left(\frac wn\right)-\frac12 = \frac wn-\frac12 + 2\frac{n}{n-1}(1-p_{11\to11}) \frac wn \left(\frac34-\frac wn\right).
\]
which is a concave quadratic, since its quadratic coefficient is negative on $[1/2, 1]$. Hence its minimum on $[1/2, 1]$ is attained at one of the endpoints. 

At $w/n = 1/2$, the value is $\frac{n}{4(n-1)}(1-p_{11\to11}) \ge \frac{1-p_{11\to11}}4$; at $w/n = 1$, the value is $\frac12-\frac{n}{2(n-1)}(1-p_{11\to11}) = \frac{np_{11\to11}-1}{2(n-1)} \ge \frac{p_{11\to11}}4$. Hence, 
\[
    E(w/n)-\frac12 \ge \min\left\{ \frac{1-p_{11\to11}}{4}, \frac{p_{11\to11}}{4} \right\} = \frac{1 - \sigma}{4}.
\]
Therefore, the event $X_{t+1}<n/2$ implies
\[
X_{t+1}-\E[X_{t+1}\mid X_t=w]
\leq -n\frac{1 - \sigma}{4}.
\]
Equivalently, $\{X_{t+1}<n/2\} \subseteq \left\{ X_{t+1}-\E[X_{t+1}\mid X_t=w] \leq -n\frac{1 - \sigma}{4} \right\}$.
Finally, applying the one-step concentration bound Lemma~\ref{lem:hoeffding-for-weight} with $\lambda =  n\frac{1 - \sigma}{4}$, we obtain the desired claim.\qedhere\\
\end{proof}

Now we are ready to control the hitting time to $[(3/4 - \delta)n, (3/4 + \delta)n]$:

\begin{proposition}\label{prop:hitting-time}
Let $\delta, m, \beta > 0$. There is a constant $c > 0$ depending on $\delta$, $m$, $\beta$, and $0 < p_{11\to11} < 1$ such that for large enough (even) $n$ we have \[
\PP[T_{hit}(w, [(3/4 - \delta)n, (3/4 + \delta)n] > c\log n] \leq \frac{1}{3^w\binom{n}{w}}n^{-m} + e^{-\beta n}
\]
for any $1 \leq w \leq n$.
\end{proposition}
\begin{proof}
Let $T = \lceil K\log n\rceil $, where $K$ is determined later. Let
$T_1, \dots, T_D$ be the times $0 \leq t < T$ when downcrossing happens, i.e., $X_t \geq n/2$ and $X_{t+1} < n/2$. Also let $T_0 = -1$ and $T_{D+1}=T$.

By Lemma~\ref{lem:rare-downcross}, for every fixed set of $d+1$ times at which such downcrossings occur, the Markov property gives \[ \PP[\text{all these $d+1$ times are downcrossings}] \leq \exp\left( -\frac{(d+1)(1 - \sigma)^2n}{512}\right). 
\] 
where $0 < 1 - \sigma < 1$. Thus \[ 
\PP[D>d] \leq \binom{T}{d+1} \exp\left( -\frac{(d+1)(1 - \sigma)^2n}{512}\right).
\] 
Since $T=O(\log n)$, for fixed $d$ and $K$, and for sufficiently large $n$, \[ 
\binom{T}{d+1} \leq \exp\left(\frac{(d+1)(1 - \sigma)^2n}{1024}\right). 
\] 
Hence \[ 
\PP[D>d] \leq \exp\left( -\frac{(d+1)(1 - \sigma)^2n}{1024}\right). 
\]
If $D \leq d$, then for some $0 \leq i \leq D$ there is a run with no downcrossing from $[n/2,n]$ to $[1,n/2)$ of length at least $\left\lfloor\frac{T}{d+1}\right\rfloor-1$. In other words, there is some run $\dots, X_{t-1}, X_t, X_{t+1}, \dots$ of length at least $\left\lfloor\frac{T}{d + 1}\right\rfloor-1$ such that the weight chain does not cross from $[n/2, n]$ to $[1, n/2)$ during the run. 

By looking at the first or last part of this long run, we either find a run of length at least $K'\log n$ (where $K'$ can be chosen arbitrarily close to $K/(2(d+1))$ for sufficiently large $n$) during which $X_t \in [1, n/2)$, or a run of length at least $K'\log n$ during which $X_t \in [n/2, n]$. In particular, during a run with no downcrossing, the chain can pass from $[1,n/2)$ to $[n/2,n]$ at most once. Here $K'$ is a constant depending on $K$ and $d$, and can be made arbitrarily large by choosing $K$ depending on $d$. We bound the two alternatives separately: a long low-region run is controlled by Proposition~\ref{prop:hitting-from-below}, while a long high-region run hits $[(3/4 - \delta)n, (3/4 + \delta)n]$ with exponentially high probability.

Suppose we have a run of length at least $K'\log n$ during which $X_t \in [n/2, n]$. We split this run into $\left\lfloor\frac{K'}{N}\log n \right\rfloor$ blocks of length $N$, where $N$ is the constant in Lemma~\ref{lem:hitting-time-narrow}. Each block has a probability of at most $2N\exp\left(-\frac{(1 - \sigma)^2\delta^2 n}{128}\right)$ of failing to hit the target interval $[(3/4 - \delta)n, (3/4 + \delta)n]$, so, after summing over the at most $T$ possible starting locations of such a run, the probability that some such run fails to hit the target interval is at most 
\[ 
T(2N)^{\frac{K'}{N}\log n} \exp\left(-\frac{(1 - \sigma)^2\delta^2 n}{128} \left\lfloor\frac{K'\log n}{N}\right\rfloor\right) \leq e^{-(\beta+1) n}
\]
after choosing $K'$ large enough.

On the other hand, we bound the probability of having a run of length at least $K'\log n$ entirely contained in $[1, n/2)$. By Proposition~\ref{prop:hitting-from-below}, applied with $2m+2$ in place of $m$, there is a $c > 0$ such that the probability, from original starting weight $w$, of having a run of length at least $c\log n$ starting at a particular time $t$ and staying in $[1, n/2)$ is at most $\frac{1}{3^w\binom{n}{w}}n^{-2m-2}$ for large $n$. 

Thus, summing over possible starting locations of the run, the probability of a trajectory of length $K\log n$ having any run of length at least $c\log n$ staying in $[1, n/2)$ is at most $\frac{K(\log n)}{3^w\binom{n}{w}}n^{-2m-2} \leq \frac{1}{3^w\binom{n}{w}}n^{-m}$ for large enough $n$. 

To summarize, a run of the weight chain starting at $w$ of length $K\log n$ can fail to hit the target interval $[(3/4 - \delta)n, (3/4 + \delta)n]$ if any of the following events happen: \begin{itemize}
    \item It downcrosses from $[n/2,n]$ to $[1,n/2)$ more than $d$ times: this happens with probability at most \[\exp\left( -\frac{(d+1)(1 - \sigma)^2n}{1024}\right).\]
    \item It fails to hit the target interval while it is in $[n/2, n]$: this happens with probability at most $e^{-(\beta+1) n}$ for sufficiently large $n$, after increasing $K$ if necessary.
    \item It spends too much consecutive time outside of $[n/2, n]$: this happens with probability at most $\frac{1}{3^w\binom{n}{w}}n^{-m}$ for sufficiently large $n$.
\end{itemize}
Now choose $d$ large enough that  \[\frac{(d+1)(1 - \sigma)^2}{1024} \ge \beta+1. \] Then \[\PP[D>d]\le e^{-(\beta+1)n} \] for sufficiently large $n$. Choose $K$ large enough that $K'>c$, where $c$ is the constant from Proposition~\ref{prop:hitting-from-below}, and large enough that the high-region run error above is at most $e^{-(\beta+1) n}$. Finally, we combine the three alternatives and obtain
\[
\PP[T_{hit}(w,[(3/4-\delta)n,(3/4+\delta)n])>T] \le \frac{1}{3^w\binom nw}n^{-m}+e^{-\beta n}.
\]
Since $T=\lceil K\log n\rceil$, this gives the desired hitting-time bound.\qedhere\\
\end{proof}

 \section*{Acknowledgements}

This work was supported by NSF Grant CCF 2338816, the Packard Fellowship for Science and Engineering, and an NSF Waterman Award
DMS-2140043. We also express our thanks to Jin Ming Koh and Professors Zongchen Chen, Ryan O'Donnell, John Preskill, and Min-Hsiu Hsieh for insightful discussions.\\

\newpage

\bibliographystyle{alpha}
\bibliography{main}

 \newpage
 
\appendix

\section{Auxiliary Propositions}\label{sec:appendix-lemmas}

\subsection{Distance Condition for Codes from Unitaries}
The following proposition from \cite{Brown_2013} is used to prove our main theorems in Subsection~\ref{sec:circuit-to-code}.

\begin{proposition}[{\cite[Proposition II.1]{Brown_2013}}]\label{"BF_II1_proof"}
    A unitary $U\in\cC_n$ defines a quantum error-correcting code of distance at least $d+1$ if and only if for all $\nu_A \in \{0,1,2,3\}^k - \{0^k\}, \nu_B \in \{0,3\}^{n-k}$, and $\mu\in\{0,1,2,3\}^n$ of weight $1\leq w(\mu)\leq d$, we have $\tr[\sigma_\mu U (\sigma_{\nu_A} \otimes \sigma_{\nu_B})U^\dagger] = 0$.
\end{proposition}
\begin{proof}
First, we prove the ``if''; suppose for all $\nu_A \in \{0,1,2,3\}^k - \{0^k\}, \nu_B \in \{0,3\}^{n-k}$, and $\mu\in\{0,1,2,3\}^n$ of weight $1\leq w(\mu)\leq d$, we have $\tr[\sigma_\mu U (\sigma_{\nu_A} \otimes \sigma_{\nu_B})U^\dagger] = 0$. 

We have
$\langle\bar x | \sigma_\mu | \bar y\rangle = \langle x|_A \otimes \langle 0|_B U^\dagger \sigma_\mu U|y\rangle_A \otimes |0\rangle_B = \tr[\sigma_\mu U(|y\rangle \langle x|\otimes |0\rangle \langle 0|)U^\dagger]$. Here, note that \begin{align*}
    |y\rangle \langle x| &= \frac{\delta_{xy}}{2^k} \sigma_0 + \frac{1}{2^k} \sum_{\nu_A\in \{0,1,2,3\}^k - \{0\}^k} \tr[\sigma_{\nu_A} |y\rangle \langle x|] \sigma_{\nu_A}, \qquad |0\rangle \langle 0|^{\otimes n-k} = \frac{1}{2^{n-k}}\sum_{\nu_B\in \{0,3\}^{n-k}} \sigma_{\nu_B}.
\end{align*}
Therefore, if $\tr[\sigma_\mu U (\sigma_{\nu_A} \otimes \sigma_{\nu_B})U^\dagger] = 0$, then $\langle \bar{x} | \sigma_\mu | \bar{y} \rangle = \tr[\sigma_\mu U(\frac{\sigma_0}{2^k} \otimes |0\rangle \langle 0|)U^\dagger] \cdot \delta_{xy} \eqqcolon C_\mu \delta_{xy}$.\\

Now, assume $U$ satisfies the condition that for all $x,y\in \{0,1\}^k$ and all $\mu\in \{0,1,2,3\}^n$ with $1\leq w(\mu)\leq d$, $\langle \bar x | \sigma_\mu | \bar y\rangle = C_\mu \delta_{xy}$. We write $\sigma_{\nu_A} = \sum_{x,y\in \{0,1\}^k} \sigma_{\nu_A} (x,y) |x\rangle \langle y|$. Thus, if $1\leq w(\mu)\leq d$, we have \begin{align*}
    \tr[\sigma_\mu U(\sigma_{\nu_A} \otimes |0\rangle \langle 0| )U^\dagger] &= \sum_{x,y\in \{0,1\}^k} \sigma_{\nu_A} (x,y) \tr[\sigma_\mu U(|x\rangle \langle y| \otimes |0\rangle \langle 0| )U^\dagger] \\
    &= \sum_{x,y \in \{0,1\}^k} \sigma_{\nu_A}(x,y) C_\mu \delta_{xy} = 0,
\end{align*}
since $\tr[\sigma_{\nu_A}] = \sum_x \sigma_{\nu_A}(x,x) = 0$. Finally, since $U$ transforms Pauli operators to Pauli operators, we have that $\tr[\sigma_\mu (U \sigma_{\nu_A} \otimes |0\rangle \langle 0|) U^\dagger] = 0$ implies that for all $\nu_B \in \{0,3\}^{n-k}$, it must be that $\tr[\sigma_\mu (U \sigma_{\nu_A} \otimes \sigma_{\nu_B}) U^\dagger] = 0$.\qedhere\\

\end{proof}

\subsection{Light Cone Lower Bound} 
In this subsection, we use a light cone argument to show that $\Omega(n\log n)$ number of gates and $\Omega(\log n)$ layers are necessary for a random matching circuit to produce a code of linear distance with probability bounded away from 0. See the discussion in Subsection~\ref{sec:further-disc}.
\begin{proposition} \label{prop: light-cone}
    Consider an ensemble of $T$ layers in which layer $r$ places arbitrary two-qubit gates on a uniformly random set of fixed size $B_r \leq  n/2$ disjoint pairs, independently of earlier layers. Let $N\coloneqq \sum_{r=1}^T B_r$ be the total number of two-qubit gates. Suppose that, with probability at least $p$, the resulting Clifford encoder defines an $[n,k]$ stabilizer code of distance at least $d+1$, where $k\geq 1$. Then
    \[
        N\geq \frac{n-1}{2}\ln\bigl(p(d+1)\bigr).
    \]
    Therefore, if $p$ is bounded below by a positive constant and $d\geq\delta n$ for some constant $\delta>0$, then $N=\Omega(n\log n)$ and $T=\Omega(\log n)$.
\end{proposition} 
\begin{proof}
     Fix a logical input qubit $i\in[k]$. Let $L_r(i)$ be the forward light cone of qubit $i$ after the first $r$ layers. In particular, let $L_0(i)=\{i\}$ and recursively define $L_{r+1}(i)$ by adding to $L_r(i)$ every qubit that is paired in layer $r+1$ with a qubit in $L_r(i)$. Let $B_r = b$ and $|L_r| = y$. There are $y(n-y)$ unordered pairs with one endpoint in $L_r(i)$ and the other endpoint in $[n]\setminus L_r(i)$.
     
     Since the matching in layer $r+1$ is a uniformly random matching of size $b$, every unordered pair is included with probability $\frac{b}{\binom{n}{2}}$. Therefore, the expected number of matching edges crossing from $L_r(i)$ to its complement is 
    \[
        b \cdot \frac{y(n-y)}{\binom{n}{2}} = b \cdot  \frac{2y(n-y)}{n(n-1)}.
    \]
    We note that each such crossing edge adds at most one new qubit to the light cone.
    Therefore
    \begin{align*} \E[|L_{r+1}| \mid |L_{r}|=y, B_r = b ] &= y + b \cdot  \frac{2y(n-y)}{n(n-1)} \\
    &\leq y\left(1+\frac{2b}{n-1}\right).
    \end{align*}
    Taking expectations and iterating over the layer, we bound the expected light cone at the final layer $T$,
    \begin{align*}\E [|L_T|] &\leq \prod_r (1+\frac{2B_r}{n-1}) \\
    &\leq \exp(\frac{2N}{n-1}).
    \end{align*}
    On the event that the circuit has distance at least $d+1$, every nontrivial encoded logical Pauli has weight at least $d+1$. Hence, with probability at least $p$, we have $|L_T(i)|\ge d+1$. To see this, note that for a Clifford encoder $U$, the encoded logical Pauli $UX_i U^\dagger$ is supported inside the forward light cone of logical qubit $i$. If the resulting stabilizer code has distance at least $d+1$, then every nontrivial logical Pauli has weight at least $d+1$. Therefore, $|L_T(i)|\geq d+1$.
    Applying Markov's inequality, we see that
    \[ p \leq \Pr[|L_T|\geq d+1]\leq \frac{\exp\left(\frac{2N}{n-1}\right)}{d+1}. \] 
    Since $d=\delta n$ and $p$ is bounded below by a positive constant,
    \begin{align*} 
        N &\geq \frac{n-1}{2}\ln\bigl(p(d+1)\bigr)
        \\  &\geq \frac{n-1}{2}\ln\bigl(p(\delta n+1)\bigr) 
        \\ &= \left(\frac{1}{2}-o(1)\right)n\ln n, 
    \end{align*}  which proves that $N=\Omega(n\log n)$.
    Finally, each layer contains at most $n/2$ two-qubit gates, so $ N\leq \frac{nT}{2}$. Thus, $ T\geq \frac{2N}{n}=\Omega(\log n)$.
\end{proof}

\end{document}